\documentclass{article}

\usepackage[utf8]{inputenc}
\usepackage[T1]{fontenc}
\usepackage{times}
\usepackage{helvet}
\usepackage{courier}

\usepackage{amsmath,amssymb,amsfonts,amsthm,mathtools}

\usepackage{booktabs}
\usepackage{array}
\usepackage{graphicx}
\usepackage{caption}
\usepackage{float}
\usepackage{placeins}

\usepackage{microtype}
\usepackage{setspace}
\usepackage{arxiv}

\usepackage[numbers,sort&compress]{natbib}
\usepackage{url}
\usepackage[
  colorlinks=true,
  linkcolor=blue,
  citecolor=blue,
  urlcolor=blue
]{hyperref}
\usepackage[noabbrev]{cleveref}

\graphicspath{{figures/}}
\numberwithin{equation}{section}

\newtheorem{theorem}{Theorem}[section]
\newtheorem{lemma}[theorem]{Lemma}
\newtheorem{corollary}[theorem]{Corollary}
\newtheorem{proposition}[theorem]{Proposition}
\theoremstyle{definition}
\newtheorem{definition}[theorem]{Definition}
\theoremstyle{remark}
\newtheorem{remark}[theorem]{Remark}

\newcommand{\E}{\mathbb{E}}
\newcommand{\Prb}{\mathbb{P}}

\newcommand{\Q}{Q}
\newcommand{\PSAPST}{P-SAPST}

\title{Degeneracy-Guided List Compression for Greedy Graph Coloring}

\author{
    Rong Fu \\
    Independent Researcher \\
    Corresponding author \and
    Yongtai Liu \\
    Independent Researcher \and
    Xiaowen Ma \\
    Independent Researcher \and
    Wangyu Wu \\
    Independent Researcher \and
    Long Zhang \\
    Independent Researcher \and
    Hongbo Zhang \\
    Independent Researcher \and
    Yangchen Zeng \\
    Independent Researcher \and
    Hoi Leong Lee \\
    Independent Researcher \and
    Hao Zhang \\
    Independent Researcher
}

\date{}
\renewcommand{\headeright}{}
\renewcommand{\undertitle}{}
\renewcommand{\shorttitle}{Degeneracy-Guided List Compression}

\hypersetup{
  pdftitle={Degeneracy-Guided List Compression for Greedy Graph Coloring},
  pdfsubject={cs.DS; math.CO; cs.DM},
  pdfauthor={Rong Fu, Yongtai Liu, Xiaowen Ma, Wangyu Wu, Long Zhang, Hongbo Zhang, Yangchen Zeng, Hoi Leong Lee, Hao Zhang},
  pdfkeywords={graph coloring; palette sparsification; degeneracy; conflict graph; sublinear algorithms; experimental algorithms},
}

\begin{document}
\maketitle

\begin{abstract}
We study degeneracy guided list compression for greedy graph coloring when graph structure is available before colors are sampled. Our exposure calibrated ordering framework assigns each vertex an independent uniform list according to its backward neighborhood in a color independent order. Its certified instantiation, Profiled Structure Aware Asymmetric Palette Sparsification, or \PSAPST, reverses a minimum degree removal sequence and obtains every backward exposure from the removal profile. For each fixed profile, we characterize the exact local budget required by independent uniform lists under history robust greedy recovery. The profile yields linear list volume on high degree forests and on a core fringe family where reciprocal rank allocation requires \(\Theta(n\log^2 n)\) sampled colors. Exact conflict expectation, concentration, and a dense exposure barrier complete the theoretical description.

The evaluation contains 40,320 runs over SAPBench and two SNAP networks. At the theorem scale, \PSAPST\ reduces mean list size by 47.6 percent relative to calibrated APST while attaining 99.8 percent observed greedy success. P-SAPST Lite replaces peeling with a degree order and provides a lower latency order choice within the same framework. On stress graphs with 250,000 vertices and up to 1,251,868 edges, Lite obtains a payload ratio of 0.865, while calibrated APST obtains 7.886. On email Enron, the corresponding ratios are 0.193 and 5.814. Compression is strongest on hub dominated and power law graphs and disappears near the dense exposure barrier. The method complements edge oblivious streaming APST by addressing an offline regime in which structural plans can be reused.
\end{abstract}


\section{Introduction}
\label{sec:introduction}

Greedy palette compression has two design coordinates: the order in which vertices are recovered and the list budget assigned to each position. Reciprocal rank allocation fixes both coordinates without inspecting the edge set, although the difficulty of a greedy step is determined by the neighbors that precede it. This topology blind choice is essential under edge oblivious access. When adjacency is already available, the same rule can allocate large lists to vertices whose recovery histories exclude only a few colors.

Palette sparsification asks how many colors from a common palette must be retained at each vertex while preserving a proper coloring. For maximum degree \(\Delta\), ordinary greedy coloring uses \(\Delta+1\) colors. The palette sparsification theorem shows that logarithmic independent lists still contain a proper coloring with high probability and supports algorithms in restricted access models \citep{assadi2019sublinear}. Later work sharpened the threshold, treated variable source palettes, and connected sparsification with broader probabilistic principles \citep{kahn2023asymptotics,kahn2024asymptotics,ashvinkumar2024palette}.

Asymmetric palette sparsification makes recovery especially simple by assigning reciprocal rank budgets and applying greedy coloring \citep{assadi2026simple}. Its topology blind design is valuable in streaming and sublinear settings. The same design leaves an offline question unresolved. If adjacency information is visible before lists are sampled, can a graph dependent order reduce both list storage and the conflict graph while retaining a high probability greedy guarantee?

We formalize this question as offline greedy palette compression. A planner may inspect a static graph and publish an order with vertex budgets, while a recovery worker receives only fresh sampled lists and the resulting conflict graph. This separation is natural when one topology serves repeated palette constrained queries, when independent list randomness must be renewed, or when the sampled certificate rather than a fixed coloring is the downstream object. A payload ratio below one means that the certificate is smaller than the input edge list under the same integer coding model. The formulation does not compete with storing one already known coloring.

The central primitive is exposure calibrated ordering. Any color independent order induces a backward exposure at every vertex, and the calibration converts that exposure into a sufficient list length. Reverse minimum degree peeling gives a certified instantiation because every recorded removal degree equals the corresponding backward exposure. \PSAPST\ uses this identity to assign short lists where the forbidden palette fraction is small, while dense regions retain larger budgets. P-SAPST Lite and the optional selector use cheaper or empirically selected orders within the same framework. Figure~\ref{fig:mechanism} contrasts the certified profile with reciprocal rank allocation.

\begin{figure}[t]
\centering
\includegraphics[width=\textwidth]{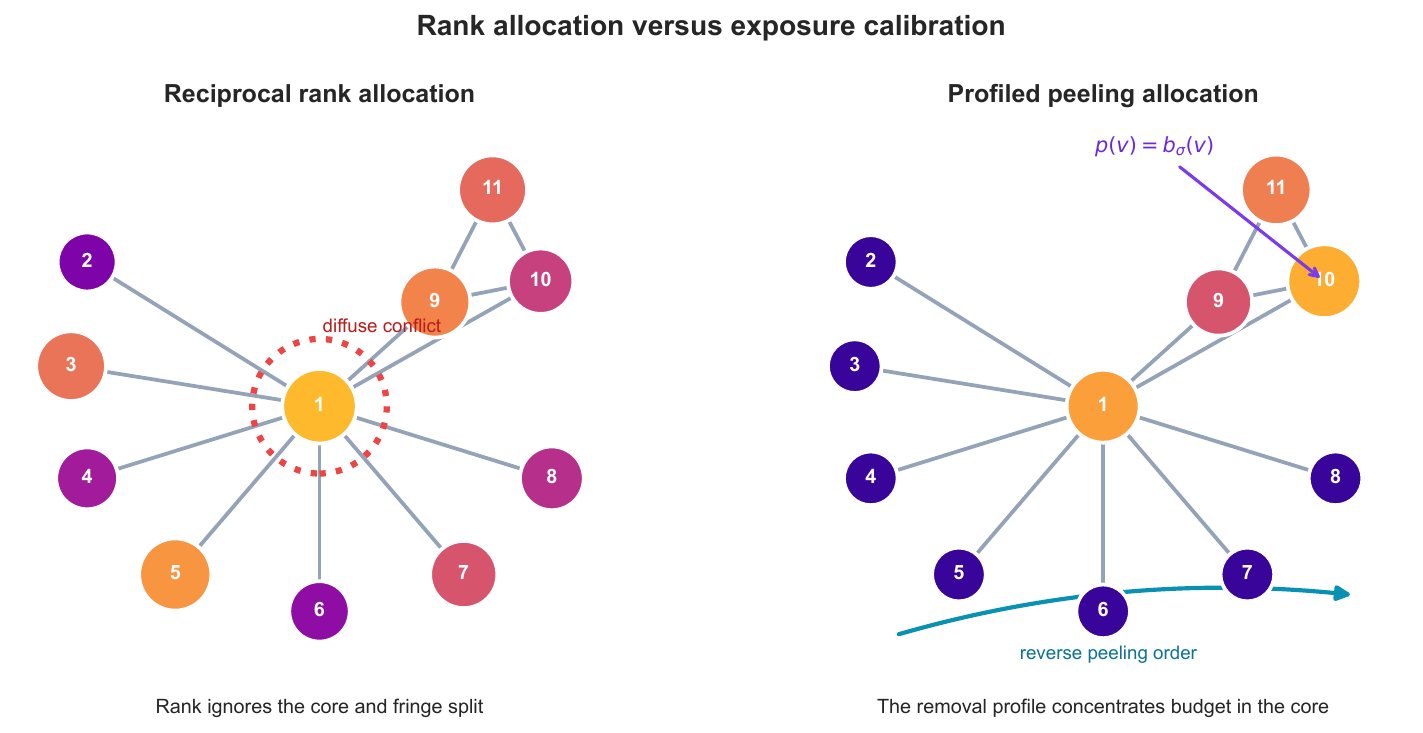}
\caption{Reciprocal rank allocation and exposure calibrated allocation on the same heterogeneous graph. The local potential is \(\phi_\sigma(v)=\log(\Q/b_\sigma(v))\). A large value means that only a small palette fraction can be forbidden, so the assigned list is short. Reverse peeling concentrates palette volume in the constrained core instead of spending it on the sparse fringe.}
\label{fig:mechanism}
\end{figure}

Our contributions are as follows. First, we define an exposure calibrated ordering framework for offline greedy palette compression. A product probability space separates planning randomness from list samples, so graph dependent orders cannot leak sampled colors into the recovery history. For a fixed exposure \(b\), we characterize the exact minimum uniform list length required by history robust greedy recovery. Summing these local prices gives an exact lower bound for that fixed profile within the independent uniform list model.

Second, we give a certified degeneracy guided instantiation. Reverse minimum degree peeling recovers the complete backward exposure profile pointwise. On forests with polynomial maximum degree, \PSAPST\ has linear total list volume and expected conflict size \(O(n/\Delta)\), while theorem scale reciprocal rank allocation uses \(\Theta(n\log^2 n)\) sampled colors. A core fringe theorem extends the separation to graphs containing a genuine dense core. We also prove exact conflict expectation, a bounded difference concentration inequality, and a dense exposure barrier that explains why regular and clique like regions resist compression.

Finally, we develop certified and low latency implementations and evaluate them on SAPBench and two SNAP networks. Full \PSAPST\ carries the pointwise profile theorem. P-SAPST Lite applies the same calibration to a degree order and inherits a conditional approximation whenever its exposure profile is dominated by the peeling profile. The study contains 40,320 runs, fixed profile lower bound audits, confidence intervals, list tails, payload ratios, phase timings, repeated query crossover measurements, and stress graphs with 250,000 vertices and more than 1.25 million edges.

\textbf{Problem setting and scope.} We target the offline regime in which adjacency is available before sampling. This regime arises when a static graph serves repeated palette constrained requests with fresh randomness or when a compressed sampled certificate is stored or communicated. The model is complementary to edge oblivious APST. Edge oblivious allocation remains preferable for a single pass over an unknown graph, while exposure calibration is useful when structural metadata can be reused. A central open problem is whether a semi streaming sketch can approximate a sufficiently accurate exposure profile without materializing the graph. Related local constraint problems, including list coloring and defective coloring, require separate recovery analyses.

\section{Related Work}
\label{sec:related-work}

\subsection{Palette sparsification}
Assadi, Chen, and Khanna proved that logarithmic random lists preserve a
proper \((\Delta+1)\) coloring with high probability
\citep{assadi2019sublinear}. Kahn and Kenney established sharp asymptotic
forms and extended the analysis to vertex dependent source palettes
\citep{kahn2023asymptotics,kahn2024asymptotics}. Ashvinkumar and Kenney
related palette sparsification to spread distributions
\citep{ashvinkumar2024palette}. Random list colorability has also been
examined from a general combinatorial viewpoint
\citep{hefetz2025coloring}.

\subsection{Structure restricted sparsification}
Additional graph structure can improve either the target palette or the
sampling requirement. Alon and Assadi studied palette sparsification beyond
the basic degree plus one setting, including larger palettes and
triangle free coloring with fewer than \(\Delta+1\) colors
\citep{alon2020palette}. Dhawan extended this direction to graphs whose
vertex neighborhoods induce few edges \citep{dhawan2024palette}. These
results exploit structural classes to improve the target number of colors
or the symmetric sampling threshold.

\subsection{Asymmetric and distributed recovery}
Assadi and Yazdanyar introduced asymmetric palette sparsification and showed
that reciprocal rank lists admit a direct greedy recovery algorithm
\citep{assadi2026simple}. Distributed palette sparsification instead asks
how the coloring can be recovered under locality and communication
constraints \citep{flin2024distributed}. Related degree plus one coloring
algorithms have been developed for congested, massively parallel, and local
models \citep{chang2019complexity,halldorsson2022near,maus2023distributed}.
\PSAPST\ does not improve their round or query bounds because it assumes
adjacency access before list sampling.

\subsection{Streaming and structural orders}
The semi streaming model permits near linear memory while edges arrive
sequentially \citep{feigenbaum2005graph}. Coloring in this model includes
degeneracy based, deterministic, and adversarially robust algorithms
\citep{bera2019graph,assadi2022deterministic,assadi2023coloring}. The
smallest last ordering of Matula and Beck supplies the classical peeling
foundation for degeneracy aware coloring \citep{matula1983smallest}.
We use the same peeling process for a different purpose: its entire removal
profile determines vertexwise random list budgets in an offline compressed
instance.

\subsection{Class restricted and profile calibrated objectives}
The structural information used by Alon and Assadi and by Dhawan changes
the colorability regime for a graph class \citep{alon2020palette,dhawan2024palette}.
Our objective is orthogonal. We retain the common palette
\(\{1,\ldots,\Delta+1\}\) and reduce the vertexwise list budgets and the
induced conflict graph. Their assumptions are expressed through properties
such as triangle freedom or neighborhood sparsity. Our parameter is the
complete exposure vector of a chosen greedy order. This distinction is
most visible on heterogeneous graphs. A high degree forest has degeneracy
one and can contain a hub of degree \(\Theta(n)\); the peeling profile then
gives linear total list volume even though the maximum degree is large.
The comparison does not imply dominance in either direction. Class
restricted sparsification can reduce the target color count, while profile
calibration compresses a fixed \((\Delta+1)\) palette certificate.

\section{Formal Model}
\label{sec:model}

\subsection{Offline greedy palette compression}

Let \(G=(V,E)\) be a simple undirected graph with \(n=|V|\), maximum degree
\(\Delta\), and common palette
\begin{equation}
[\Q]=\{1,\ldots,\Q\},
\qquad
\Q=\Delta+1.
\label{eq:palette}
\end{equation}
where \([\Q]\) denotes the colors available to every vertex.

A planning rule receives \(G\) and returns an order
\(\sigma=(v_1,\ldots,v_n)\) together with integer budgets
\(\ell(v)\in[\Q]\). After this plan is fixed, each vertex independently
samples a uniformly random \(\ell(v)\)-subset \(L(v)\) of \([\Q]\).
Greedy recovery processes \(\sigma\) from left to right and assigns an
arbitrary color in \(L(v_i)\) that has not been used by an earlier neighbor.

For the sampled lists, define the conflict graph
\begin{equation}
G_L=(V,E_L),
\qquad
E_L=
\bigl\{
\{u,v\}\in E:L(u)\cap L(v)\neq\varnothing
\bigr\}.
\label{eq:conflict-graph}
\end{equation}
where \(E_L\) contains exactly the input edges that remain relevant after
list sampling.

\begin{definition}[Offline greedy palette compression]
\label{def:offline-compression}
For a target failure probability \(\delta\), the planning problem is to
choose \((\sigma,\ell)\) before list sampling so that greedy recovery fails
with probability at most \(\delta\). Among feasible plans, the compression
criteria are
\begin{align}
S(\ell)
&=
\sum_{v\in V}\ell(v),
\label{eq:objective-storage}\\
C(\ell)
&=
\E\lvert E_L\rvert.
\label{eq:objective-conflict}
\end{align}
where \(S(\ell)\) is total list volume and \(C(\ell)\) is the expected
number of retained conflict edges.
\end{definition}

The definition compresses a palette constrained recovery certificate. It
does not claim to encode one already known coloring with minimum space.

If recovery fails, the sampled certificate can be rejected and the lists
can be drawn again from the same fixed plan. Since one trial succeeds with
probability at least \(1-\delta\), independent resampling requires at most
\begin{equation}
\E A
\leq
\frac{1}{1-\delta}
\label{eq:expected-resampling}
\end{equation}
where \(A\) is the number of independently sampled certificates examined
before success. This fallback preserves the palette and the planning guarantee,
although it changes a single sample request into a Las Vegas service.

\subsection{Repeated palette queries}

Suppose a fixed graph supports \(r\) independent list sampling and recovery tasks. For a method \(M\), let \(P_M\) denote the cost of constructing its structural plan and let \(q_M\) denote the expected cost of one subsequent budget, sampling, conflict construction, and recovery phase. The total service cost is
\begin{equation}
T_M(r)=P_M+r q_M.
\label{eq:repeated-query-cost}
\end{equation}
where \(P_M\) is paid once and each query draws fresh independent lists. This model does not assume that a sampled coloring is reused. It isolates the setting in which graph metadata remains fixed while palette constrained requests change.

\subsection{Randomness and exposure history}

The probability space is a product
\begin{equation}
\Omega=\Omega_{\mathrm{aux}}\times\Omega_{\mathrm{list}}.
\label{eq:product-space}
\end{equation}
where \(\Omega_{\mathrm{aux}}\) contains random tie breaking used by the
planner and \(\Omega_{\mathrm{list}}\) contains the independent vertex
lists. The order and budgets are measurable with respect to
\((G,\Omega_{\mathrm{aux}})\) only.

For a fixed realization of the plan, let \(\mathcal{H}_i\) contain the lists
and greedy choices revealed before \(v_i\). Then
\begin{equation}
L(v_i)\ \perp\!\!\!\perp\ \mathcal{H}_i
\mid
(G,\Omega_{\mathrm{aux}}).
\label{eq:history-independence}
\end{equation}
where the conditional independence states that the current list is not
contaminated by the earlier recovery history. The planner may inspect all
graph structure, but it may not inspect sampled colors.

For a vertex \(v_i\), define
\begin{align}
B_\sigma(v_i)
&=
N(v_i)\cap\{v_1,\ldots,v_{i-1}\},
\label{eq:backward-neighborhood}\\
b_\sigma(v_i)
&=
\lvert B_\sigma(v_i)\rvert.
\label{eq:backward-degree}
\end{align}
where \(B_\sigma(v_i)\) is the earlier neighborhood and
\(b_\sigma(v_i)\) is the backward exposure. The number of distinct
forbidden colors at \(v_i\) is at most \(b_\sigma(v_i)\), regardless of how
large the earlier neighbors' lists were.

\section{Exposure-Calibrated List Compression}
\label{sec:methodology}

\subsection{Local exposure calibration}

The local exposure potential of a positive backward degree is
\begin{equation}
\phi_\sigma(v)
=
\log\!\left(
\frac{\Q}{b_\sigma(v)}
\right).
\label{eq:local-exposure-potential}
\end{equation}
where \(\phi_\sigma(v)\) measures the logarithmic separation between the
palette and the largest forbidden set at \(v\).

For a local error allowance \(\eta\in(0,1)\), define
\begin{equation}
\lambda_\eta(b)
=
\begin{cases}
1, & b=0,\\[2pt]
\displaystyle
\min\!\left\{
\Q,
\left\lceil
\frac{\log(1/\eta)}
{\log(\Q/b)}
\right\rceil
\right\},
& 1\leq b\leq \Q-1.
\end{cases}
\label{eq:local-calibration}
\end{equation}
where \(b\) is a backward exposure and \(\lambda_\eta(b)\) is the assigned
uniform list size.

\begin{lemma}[Color independent exposure calibration]
\label{lem:local-calibration}
Fix a plan satisfying \cref{eq:history-independence}. If vertex \(v\)
receives a uniform list of size
\begin{equation}
\ell(v)=\lambda_{\eta_v}\!\left(b_\sigma(v)\right),
\label{eq:vertex-calibration}
\end{equation}
where the local allowances satisfy
\(\sum_{v\in V}\eta_v\leq\delta\), then greedy recovery succeeds with
probability at least \(1-\delta\).
\end{lemma}

The theorem setting uses \(\eta_v=\delta/n\). The proof conditions on the
entire earlier history. At most \(b_\sigma(v)\) distinct colors are
forbidden, and the current list remains independent of that set.
Setting \(\delta=n^{-c}\), where \(c>0\), gives success probability at least
\(1-n^{-c}\). Thus the recovery statement is already a high probability
guarantee over the joint list sample, rather than an expectation over
successful vertices.

\subsection{Exposure-calibrated ordering framework}

For any color independent order \(\sigma\) and local allowances
\(\{\eta_v\}_{v\in V}\), define
\begin{align}
S_{\boldsymbol{\eta}}(\sigma)
&=
\sum_{v\in V}
\lambda_{\eta_v}\!\left(b_\sigma(v)\right),
\label{eq:order-budget-objective}\\
C_{\boldsymbol{\eta}}(\sigma)
&=
\sum_{\{u,v\}\in E}
\rho_{\Q}\!\left(
\lambda_{\eta_u}(b_\sigma(u)),
\lambda_{\eta_v}(b_\sigma(v))
\right).
\label{eq:order-conflict-objective}
\end{align}
where \(S_{\boldsymbol{\eta}}(\sigma)\) is the sufficient list volume and
\(C_{\boldsymbol{\eta}}(\sigma)\) is the exact expected conflict size.
The first quantity is defined by local exposure calibration, and the
second follows from \cref{eq:intersection-probability}. Reverse peeling,
degree order, core order, and the optional selector differ only in how
they choose \(\sigma\). The role of peeling is certification rather than
universal optimization: it exposes a removal profile with a pointwise
identity, while cheaper orders may attain a smaller empirical objective
on a particular graph.

\subsection{Certified peeling profile}

Let \(r_1,\ldots,r_n\) be a minimum degree removal sequence. Immediately
before \(r_i\) is removed, define
\begin{equation}
p(r_i)
=
\deg_{G[\{r_i,\ldots,r_n\}]}(r_i).
\label{eq:removal-profile}
\end{equation}
where \(p(v)\) is the degree of \(v\) in the residual graph at its removal
time.

\begin{theorem}[Profile identity]
\label{thm:profile-identity}
For the reverse order
\begin{equation}
\sigma_{\mathrm{peel}}=(r_n,r_{n-1},\ldots,r_1),
\label{eq:reverse-peeling-order}
\end{equation}
where the sequence is read from the last removed vertex to the first,
every vertex satisfies
\begin{equation}
b_{\sigma_{\mathrm{peel}}}(v)=p(v).
\label{eq:profile-identity}
\end{equation}
where the backward exposure equals the recorded removal degree pointwise.
\end{theorem}

\PSAPST\ applies \cref{eq:local-calibration} to this complete profile with
\(\eta_v=\delta/n\). Consequently, its total list volume is
\begin{equation}
S_{\mathrm{P}}
=
\sum_{v\in V}
\lambda_{\delta/n}\!\left(p(v)\right).
\label{eq:profile-total-budget}
\end{equation}
where \(S_{\mathrm{P}}\) retains every local removal degree instead of
replacing the profile by its maximum.

If \(k=\max_v p(v)\) is the degeneracy, then
\begin{equation}
S_{\mathrm{P}}
\leq
n\lambda_{\delta/n}(k).
\label{eq:degeneracy-corollary}
\end{equation}
where the right side is the coarser degeneracy only bound. Since
\(\operatorname{degeneracy}(G)\leq
2\operatorname{arboricity}(G)-1\), the same expression also gives an
arboricity consequence.

\begin{theorem}[Sparse profile separation]
\label{thm:forest-separation}
Let \(G\) be a forest with \(\Delta\geq n^\gamma\), where \(\gamma>0\), and
let \(\delta=n^{-c}\), where \(c>0\). Then
\begin{align}
S_{\mathrm{P}}
&\leq
n\left\lceil\frac{1+c}{\gamma}\right\rceil,
\label{eq:forest-budget}\\
\E\lvert E_L\rvert
&\leq
\frac{n}{\Delta+1}\,
\left\lceil\frac{1+c}{\gamma}\right\rceil^2.
\label{eq:forest-conflicts}
\end{align}
where the first line is linear list volume and the second is
\(O(n/\Delta)\) expected conflict size. For the reciprocal rank rule
\(\ell_R(j)=\min\{\Q,\lceil c_0n\log n/j\rceil\}\), with any fixed
\(c_0>0\), the total list volume is
\begin{equation}
\sum_{j=1}^{n}\ell_R(j)=\Theta(n\log^2 n).
\label{eq:rank-forest-budget}
\end{equation}
where the asymptotic statement holds under the same polynomial maximum
degree assumption.
\end{theorem}

The theorem gives a strict structural separation. A star and a hub forest
have large maximum degree but a removal profile bounded by one. A regular
graph can have the same maximum degree and a profile close to the palette
size, so no analogous saving follows.

\begin{theorem}[Core fringe separation]
\label{thm:core-fringe-separation}
Let \(G_{s,t}\) consist of a clique on \(s\geq2\) core vertices and \(t\) pendant vertices adjacent to each core vertex. Let
\begin{equation}
n=s(t+1),
\qquad
t\geq s n^\gamma,
\qquad
\delta=n^{-c},
\label{eq:core-fringe-conditions}
\end{equation}
where \(\gamma,c>0\) are constants. Define
\begin{equation}
C_{\gamma,c}
=
\left\lceil
\frac{1+c}{\gamma}
\right\rceil.
\label{eq:core-fringe-constant}
\end{equation}
where \(C_{\gamma,c}\) is independent of \(n\). Then \PSAPST\ satisfies
\begin{align}
S_{\mathrm{P}}
&\leq
C_{\gamma,c}n,
\label{eq:core-fringe-budget}\\
\E\lvert E_L\rvert
&\leq
C_{\gamma,c}^2s.
\label{eq:core-fringe-conflict}
\end{align}
where the bounds include both the clique core and its sparse fringe. The reciprocal rank rule from \cref{thm:forest-separation} has total list volume
\begin{equation}
\sum_{j=1}^{n}\ell_R(j)=\Theta(n\log^2 n).
\label{eq:core-fringe-rank-budget}
\end{equation}
where the asymptotic separation holds for the same family.
\end{theorem}

The core fringe result is not a forest artifact. Peeling first removes every pendant vertex at exposure one and then processes the clique at exposure at most \(s-1\). The condition \(t/s\geq n^\gamma\) leaves a polynomial gap between both exposure scales and the palette, while reciprocal rank allocation cannot use that gap.

\subsection{Local necessity in the independent uniform model}

\begin{definition}[History robust local rule]
\label{def:history-robust}
A uniform list of size \(\ell\) is history robust for parameters
\((b,\eta)\) if, for every forbidden set \(F\subseteq[\Q]\) with
\(|F|=b\) that is independent of the current list, the probability that
the list is contained in \(F\) is at most \(\eta\).
\end{definition}

\begin{definition}[Exact local optimum]
\label{def:exact-local-optimum}
For \(b\in\{0,\ldots,\Q-1\}\), define
\begin{equation}
h_\eta(b)
=
\min\!\left\{
\ell\in[\Q]:
\frac{\binom{b}{\ell}}{\binom{\Q}{\ell}}
\leq
\eta
\right\},
\label{eq:exact-local-optimum}
\end{equation}
where \(\binom{b}{\ell}=0\) when \(\ell>b\), and \(h_\eta(0)=1\).
\end{definition}

\begin{theorem}[Exact history robust optimum]
\label{thm:exact-local-optimum}
A uniform list of size \(\ell\) is history robust for \((b,\eta)\) if and
only if
\begin{equation}
\ell\geq h_\eta(b).
\label{eq:exact-optimum-condition}
\end{equation}
where the comparison is exact within the independent uniform list model.
Consequently, for a fixed order and local allowances \(\eta_v\), every
history robust vertexwise plan has total budget at least
\begin{equation}
H(\sigma)
=
\sum_{v\in V}
h_{\eta_v}\!\left(b_\sigma(v)\right).
\label{eq:profile-optimum}
\end{equation}
where \(H(\sigma)\) is the exact profile lower bound.
\end{theorem}

\begin{theorem}[Analytic near optimality]
\label{thm:local-lower}
For a history robust rule with \(\ell\leq b/2\),
\begin{equation}
\ell
\geq
\frac{\log(1/\eta)}
{\log(2\Q/b)}.
\label{eq:local-lower-bound}
\end{equation}
where the denominator reflects the largest possible factor between
sampling with and without replacement. If \(b\leq\Q/2\) and the sufficient
budget in \cref{eq:local-calibration} is at most \(b/2\), then the
\PSAPST\ budget is at most twice this lower bound, apart from the ceiling.
\end{theorem}

\begin{remark}[Scope of the lower bound]
\label{rem:lower-bound-scope}
\Cref{thm:exact-local-optimum,thm:local-lower} characterize fixed
exposure profiles for independent uniform lists under history robust
greedy recovery. They do not establish optimality over correlated vertex
lists, nonuniform color distributions, or recovery algorithms that revise
earlier choices. Such extensions break the conditional product structure
used by the proof and remain open.
\end{remark}

For \(\beta\in(0,1)\), define the dense exposure set
\begin{equation}
H_\beta
=
\bigl\{
v\in V:b_\sigma(v)\geq\beta\Q
\bigr\}.
\label{eq:dense-exposure-set}
\end{equation}
where \(H_\beta\) contains vertices whose earlier neighborhood covers a
constant fraction of the palette in the worst case.

\begin{corollary}[Dense exposure barrier]
\label{cor:dense-barrier}
Suppose every vertex uses a history robust uniform list with local error at
most \(\eta\). Define
\begin{equation}
s_\beta
=
\min\!\left\{
\left\lfloor\frac{\beta\Q}{2}\right\rfloor+1,
\left\lceil
\frac{\log(1/\eta)}{\log(2/\beta)}
\right\rceil
\right\}.
\label{eq:dense-minimum-list}
\end{equation}
where \(s_\beta\) is the minimum list burden forced by dense exposure. Then
\begin{align}
\sum_{v\in H_\beta}\ell(v)
&\geq
s_\beta\lvert H_\beta\rvert,
\label{eq:dense-storage-lower}\\
\E\lvert E_L\rvert
&\geq
\lvert E(H_\beta)\rvert\,
\rho_{\Q}(s_\beta,s_\beta).
\label{eq:dense-conflict-lower}
\end{align}
where \(E(H_\beta)\) is the edge set induced by \(H_\beta\), and
\(\rho_{\Q}\) is defined in \cref{eq:intersection-probability}.
\end{corollary}

\subsection{Conflict graph theory}

For two independent uniform lists of sizes \(x\) and \(y\), define
\begin{equation}
\rho_{\Q}(x,y)
=
\begin{cases}
\displaystyle
1-\frac{\binom{\Q-x}{y}}{\binom{\Q}{y}},
& x+y\leq\Q,\\[8pt]
1,
& x+y>\Q.
\end{cases}
\label{eq:intersection-probability}
\end{equation}
where \(\rho_{\Q}(x,y)\) is their exact intersection probability.
Figure~\ref{fig:hypergeometric} displays its nonlinear dependence on both
endpoint budgets.

\begin{figure}[t]
\centering
\includegraphics[width=0.67\textwidth]{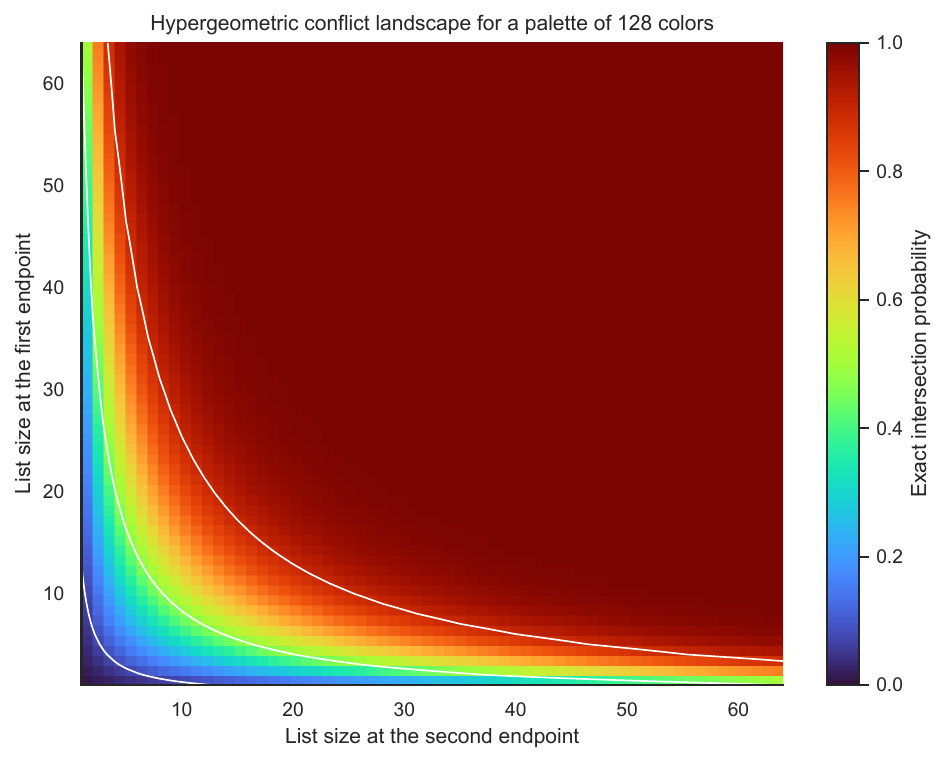}
\caption{Exact list intersection probability \(\rho_{128}(x,y)\) for two independent uniform lists. Color intensity gives the probability that endpoint lists of sizes \(x\) and \(y\) intersect. White contours mark probabilities \(0.1\), \(0.5\), and \(0.9\). The probability grows approximately with \(xy/128\) for small lists and equals one when \(x+y>128\), which explains why reducing both endpoint budgets can sharply thin the conflict graph.}
\label{fig:hypergeometric}
\end{figure}

\begin{theorem}[Conflict expectation and concentration]
\label{thm:conflict-theory}
Conditioned on any valid plan,
\begin{equation}
\E\lvert E_L\rvert
=
\sum_{\{u,v\}\in E}
\rho_{\Q}\!\left(\ell(u),\ell(v)\right).
\label{eq:exact-conflict-expectation}
\end{equation}
where the expectation is over independent uniform vertex lists. Moreover,
\begin{equation}
\Prb\!\left(
\left|
\lvert E_L\rvert-\E\lvert E_L\rvert
\right|\geq t
\right)
\leq
2\exp\!\left(
-\frac{2t^2}{\sum_{v\in V}d(v)^2}
\right).
\label{eq:conflict-concentration}
\end{equation}
where \(d(v)\) is the degree of \(v\) in the input graph.
\end{theorem}

The expectation also obeys
\begin{equation}
\E\lvert E_L\rvert
\leq
\min\!\left\{
m,
\frac{1}{\Q}
\sum_{\{u,v\}\in E}\ell(u)\ell(v)
\right\}.
\label{eq:conflict-upper}
\end{equation}
where \(m=|E|\). The bound translates a sparse exposure profile into a
small expected conflict graph whenever endpoint list products remain below
the palette size.

\subsection{Algorithm workflow}
\label{sec:algorithm-workflow}

\begin{figure}[t]
\centering
\includegraphics[width=\textwidth]{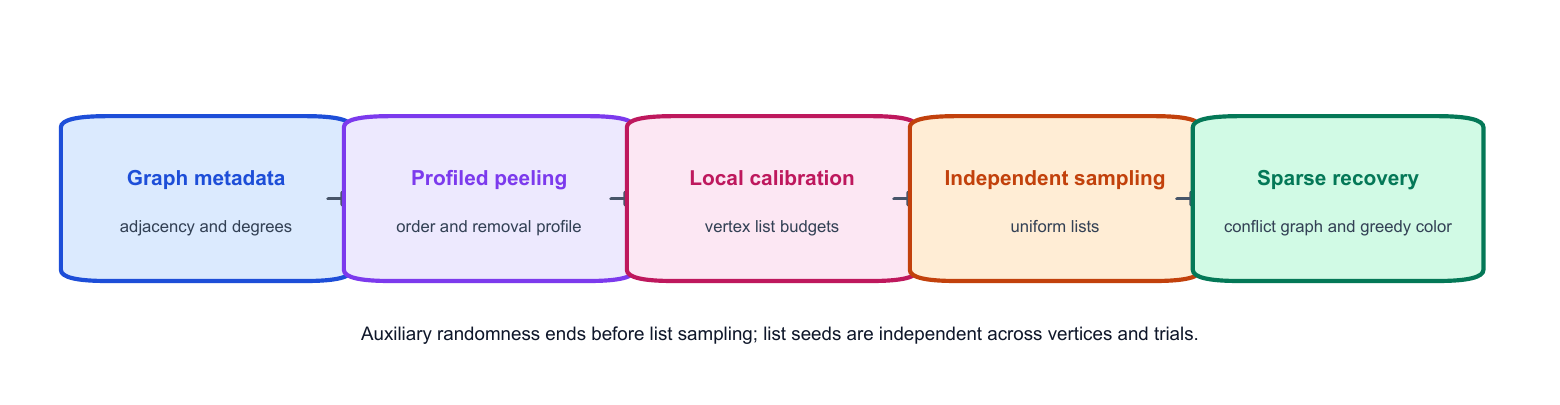}
\caption{Workflow of the certified \PSAPST\ instantiation. Structural planning fixes the order and exposure profile before independent list sampling begins. Degree, core, and selector variants retain the same calibration and replace only the order construction stage.}
\label{fig:algorithm-pipeline}
\end{figure}

\paragraph{Profile construction.}
\PSAPST\ repeatedly removes a current minimum degree vertex and records its
residual degree. Random tie breaking is drawn only from
\(\Omega_{\mathrm{aux}}\). Reversing the removal sequence produces the
greedy order and, by \cref{thm:profile-identity}, the exact backward
exposure of every vertex.

\paragraph{Lite construction.}
P-SAPST Lite uses the same exposure calibration after sorting vertices by
descending degree. It avoids heap based peeling and therefore does not
satisfy the profile identity in \cref{thm:profile-identity}. It is a low
latency implementation choice rather than a replacement for the theorem
bearing profile order.

\begin{proposition}[Conditional Lite transfer]
\label{prop:lite-transfer}
Let \(b_{\mathrm{L}}(v)\) be the backward exposure under the P-SAPST Lite order. Fix \(\theta\in(0,1]\). If \(b_{\mathrm{L}}(v)=0\) whenever \(p(v)=0\), and every remaining vertex satisfies
\begin{equation}
b_{\mathrm{L}}(v)
\leq
\Q^{1-\theta}p(v)^\theta,
\label{eq:lite-exposure-dominance}
\end{equation}
where \(p(v)\) is the peeling profile, then
\begin{equation}
\lambda_\eta\!\left(b_{\mathrm{L}}(v)\right)
\leq
\left\lceil
\frac{\lambda_\eta(p(v))}{\theta}
\right\rceil.
\label{eq:lite-budget-transfer}
\end{equation}
where the same local error allowance \(\eta\) is used by both orders.
\end{proposition}

The proposition supplies a sufficient condition, not a universal guarantee for degree ordering. Adversarial graphs can violate \cref{eq:lite-exposure-dominance} for every constant \(\theta\), so Lite has no unconditional approximation to the peeling budget. P-SAPST Lite is a low latency framework instance, while full \PSAPST\ is the certified instance with an unconditional profile identity.

\paragraph{Budget assignment.}
The method sets
\begin{equation}
\ell(v)
=
\lambda_{\delta/n}\!\left(p(v)\right).
\label{eq:psapst-budget}
\end{equation}
where \(p(v)\) is the removal profile. This phase uses no sampled color.

\paragraph{Sampling and recovery.}
Each vertex independently samples \(\ell(v)\) colors without replacement.
The conflict graph is obtained by retaining input edges whose lists
intersect, and greedy recovery follows the reverse peeling order.

\paragraph{Selector extension.}
The experiments also report a more expensive extension that compares
degree, core, and peeling orders. It minimizes the pair
\begin{equation}
\left(
\sum_{\{u,v\}\in E}
\rho_{\Q}\!\left(\ell(u),\ell(v)\right),
\sum_{v\in V}\ell(v)
\right)
\label{eq:selector-objective}
\end{equation}
where the first coordinate is exact expected conflict size and the second
is total list volume. This selector is an empirical extension, not the
method required by the profile theorems.

\subsection{Computational cost}

\begin{proposition}[Implementation complexity]
\label{prop:complexity}
Heap based profile construction requires \(O((n+m)\log n)\) time and
\(O(n+m)\) memory. Budget computation takes \(O(n+m)\) time, list sampling
takes expected \(O(S_{\mathrm{P}})\) time, and greedy recovery takes
expected \(O(m+S_{\mathrm{P}})\) time. With hashed lists, direct conflict
construction takes
\begin{equation}
O\!\left(
\sum_{\{u,v\}\in E}
\min\{\ell(u),\ell(v)\}
\right).
\label{eq:conflict-construction-cost}
\end{equation}
where the sum is the cost of intersecting endpoint lists through the
smaller set.
\end{proposition}

\begin{proposition}[Repeated query crossover]
\label{prop:repeated-query-crossover}
Let \(M\) and \(R\) be two methods with costs defined by \cref{eq:repeated-query-cost}. If \(q_M<q_R\), then
\begin{equation}
T_M(r)<T_R(r)
\quad\text{whenever}\quad
r>
\frac{(P_M-P_R)_+}{q_R-q_M}.
\label{eq:repeated-query-threshold}
\end{equation}
where \((x)_+=\max\{x,0\}\). Thus a more expensive structural plan has an explicit finite crossover whenever it reduces the subsequent query cost.
\end{proposition}

\section{Experimental Evaluation}
\label{sec:experiments}

\subsection{Questions and named graph collections}

The experiments examine list compression, greedy success, conflict graph size, statistical agreement with the exact expectation, and the time spent in each algorithm phase. They also test whether the benefit follows the gap between degeneracy and maximum degree, how closely each structural plan approaches the exact history robust profile optimum, and when preprocessing is amortized across repeated list samples.

We name the synthetic collection SAPBench. SAPBench Core contains cliques,
disjoint clique unions, Erd\H{o}s and R\'enyi graphs
\citep{erdds1959random}, random regular graphs, Barab\'asi and Albert graphs
\citep{barabasi1999emergence}, stochastic block models
\citep{holland1983stochastic}, regular bipartite graphs, unions of three
random spanning trees, rectangular grids, and a core fringe construction.
The core fringe family joins a clique of size approximately \(\sqrt n\) to
a sparse exterior. Synthetic core instances use
\begin{equation}
n\in\{128,256,512\}.
\label{eq:core-sizes}
\end{equation}
where the three sizes expose finite size behavior. Four small networks
distributed with NetworkX supplement these graphs
\citep{hagberg2007exploring}.

SAPBench Scale contains eight sparse families. It uses the random, regular, power law, block, bipartite, bounded arboricity, and grid families above, together with a hub forest. A hub forest partitions leaves among approximately \(\sqrt n\) centers and provides the separation described by \cref{thm:forest-separation}. The scale sizes are
\begin{equation}
n\in\{1{,}000,10{,}000,50{,}000\}.
\label{eq:scale-sizes}
\end{equation}
where the largest instances test list and conflict construction beyond the toy graph regime.

SAPBench Stress isolates four sparse families. It contains Erd\H{o}s and R\'enyi, Barab\'asi and Albert, bounded arboricity, and hub forest instances at
\begin{equation}
n\in\{100{,}000,250{,}000\}.
\label{eq:stress-sizes}
\end{equation}
where the largest graph contains \(1{,}251{,}868\) edges. This suite compares calibrated APST, P-SAPST Lite, and full \PSAPST.

The external suite contains the SNAP collaboration network ca HepTh and the email Enron communication network \citep{leskovec2016snap}. The graphs have 9,877 and 36,692 vertices, respectively. These data are evaluated separately from SAPBench so that synthetic scale results and observed network results remain distinguishable.

\subsection{Methods and trial structure}

Uniform list greedy uses one common logarithmic list size and a random greedy order. It is not labeled as the palette sparsification theorem because that theorem does not require greedy recovery. Proof APST implements the reciprocal rank rule with its theorem constant \(40\). Calibrated APST uses the same engineering scale as the remaining methods. P-SAPST Lite orders vertices by descending degree and applies the local exposure calibration without peeling. Core aware uses descending core number. Full \PSAPST\ uses the reverse peeling profile and therefore carries the profile identity from \cref{thm:profile-identity}. The P-SAPST selector is the optional three order extension from \cref{eq:selector-objective}.

For SAPBench Core, the scale grid is
\begin{equation}
\alpha\in\{0.50,0.75,1.00,1.25,1.50\},
\qquad
\delta=0.05.
\label{eq:experimental-grid}
\end{equation}
where \(\alpha=1\) is the theorem scale. Every graph, method, and scale uses
five independent auxiliary plans and six independent list samples for each
plan. Thus every configuration has thirty list trials. SAPBench Scale fixes
\(\alpha=1\), uses one auxiliary plan, and draws thirty independent list
samples.

There are 34 core instances and 24 scale instances. The main record count is
\begin{equation}
34\cdot7\cdot5\cdot5\cdot6
+
24\cdot6\cdot30
=
40{,}020.
\label{eq:main-run-count}
\end{equation}
where the factors represent instances, methods, scales, auxiliary plans, and list trials in their respective suites. SAPBench Stress contributes
\begin{equation}
8\cdot3\cdot10=240
\label{eq:stress-run-count}
\end{equation}
where the factors denote stress instances, methods, and independent list trials. The SAPBench subtotal is therefore \(40{,}260\) runs.

The two SNAP graphs contribute
\begin{equation}
2\cdot3\cdot10=60
\label{eq:real-run-count}
\end{equation}
where the factors denote networks, methods, and independent list trials. The complete evaluation therefore contains \(40{,}320\) runs.

\subsection{Metrics}

The total and maximum list volumes are
\begin{align}
S_L
&=
\sum_{v\in V}\lvert L(v)\rvert,
\label{eq:metric-total-list}\\
M_L
&=
\max_{v\in V}\lvert L(v)\rvert.
\label{eq:metric-max-list}
\end{align}
where \(S_L\) measures aggregate storage and \(M_L\) measures the largest
local burden. We also report the vertexwise 95th and 99th percentiles for
each plan. Greedy success means that no vertex requires a color outside its
sampled list.

The compact payload ratio is
\begin{equation}
R_{\mathrm{payload}}
=
\frac{
S_L\lceil\log_2\Q\rceil
+
2\lvert E_L\rvert\lceil\log_2 n\rceil
}{
2m\lceil\log_2 n\rceil
}.
\label{eq:payload-ratio}
\end{equation}
where the numerator stores sampled colors and conflict edge endpoints, and
the denominator stores the input edge list under the same integer coding
model. This idealized representation isolates algorithmic payload from
language and container overhead. It is not a prediction of resident Python
memory, and an optimized implementation may encode both adjacency and
conflict edges in compressed sparse formats.

For success curves and success matched comparisons, we report two sided
95 percent Wilson intervals. Running time is divided into order
construction, budget computation, list sampling, conflict construction,
and greedy recovery. Timings are wall clock measurements from a single
Python process on a 12th Gen Intel Core i5-12500H and do not include
parallel speedup. All figures are
generated directly by the experiment project without background grid
lines.

\subsection{Core results}

\begin{table}[H]
\centering
\caption{SAPBench Core at \(\alpha=1\). The relative column measures the mean list reduction from calibrated APST. Conflict is the retained edge fraction, and payload is defined by \cref{eq:payload-ratio}.}
\label{tab:main-results}
\begingroup
\setlength{\tabcolsep}{4pt}
\begin{tabular}{lrrrrrr}
\toprule
Method & Mean & vs APST & Conflict & Success & Payload & ms \\
\midrule
Uniform list greedy & 8.03 & 47.0\% & 0.913 & 90.7\% & 1.682 & 18.55 \\
Proof APST & 41.25 & -172.3\% & 1.000 & 100.0\% & 3.285 & 54.73 \\
Calibrated APST & 15.15 & 0.0\% & 0.977 & 100.0\% & 2.375 & 32.16 \\
P-SAPST Lite & 7.70 & 49.2\% & 0.746 & 99.6\% & 1.210 & 24.06 \\
Core aware & 7.70 & 49.2\% & 0.747 & 100.0\% & 1.211 & 24.75 \\
P-SAPST & 7.94 & 47.6\% & 0.788 & 99.8\% & 1.271 & 29.74 \\
P-SAPST selector & 7.69 & 49.2\% & 0.743 & 99.8\% & 1.206 & 67.91 \\
\bottomrule
\end{tabular}

\endgroup
\end{table}

\begin{table}[H]
\centering
\caption{Vertexwise list tail statistics on SAPBench Core at \(\alpha=1\). Each entry is averaged over graph instances and auxiliary plans.}
\label{tab:list-tail-results}
\begin{tabular}{lrrr}
\toprule
Method & P95 & P99 & Mean maximum \\
\midrule
Uniform list greedy & 8.03 & 8.03 & 8.03 \\
Proof APST & 43.76 & 43.76 & 43.76 \\
Calibrated APST & 31.50 & 43.76 & 43.76 \\
P-SAPST Lite & 20.18 & 33.37 & 33.78 \\
Core aware & 20.18 & 33.36 & 33.88 \\
P-SAPST & 19.56 & 32.59 & 32.74 \\
P-SAPST selector & 20.13 & 33.34 & 33.80 \\
\bottomrule
\end{tabular}

\end{table}

Table~\ref{tab:main-results} shows that \PSAPST\ uses 7.94 colors per vertex
at the theorem scale, while calibrated APST uses 15.15. This is a 47.6
percent reduction. The mean retained edge fraction decreases from 0.977 to
0.788, and the payload ratio decreases from 2.375 to 1.271. The observed
greedy success rate is 99.8 percent. The average 99th percentile list size
also falls from 43.76 to 32.59, so the gain is not produced only by reducing
small lists on the fringe.

P-SAPST Lite and core aware plans are slightly smaller in aggregate, with mean list size 7.70 and conflict fraction about 0.747. Their advantage reflects the order dependent exposure distribution: a degree or core order can produce a smaller sufficient budget on a particular heterogeneous graph even though it has no pointwise peeling identity. Full \PSAPST\ instead supplies a certified profile for every graph and an exact audit against the lower bound induced by that profile. Reverse peeling is therefore not claimed to minimize \cref{eq:order-budget-objective,eq:order-conflict-objective}. The selector reaches 7.69 colors per vertex and conflict fraction 0.743, but its mean running time is 67.91 ms, compared with 29.74 ms for \PSAPST.

\begin{table}[H]
\centering
\caption{Smallest tested scale whose Wilson lower confidence limit reaches
95 percent. ``Not attained'' means that no tested scale reaches the target.}
\label{tab:success-matched}
\begingroup
\setlength{\tabcolsep}{4pt}
\begin{tabular}{lrrrr}
\toprule
Method & Scale with lower CI at 95\% & Lower CI & List/vertex & Conflict \\
\midrule
Uniform list greedy & not attained & 89.3\% & 10.44 & 0.955 \\
Proof APST & 0.50 & 99.6\% & 41.25 & 1.000 \\
Calibrated APST & 0.75 & 99.1\% & 13.15 & 0.961 \\
P-SAPST Lite & 0.75 & 97.2\% & 6.42 & 0.674 \\
Core aware & 0.75 & 96.6\% & 6.42 & 0.675 \\
P-SAPST & 0.75 & 96.8\% & 6.65 & 0.718 \\
P-SAPST selector & 0.75 & 96.5\% & 6.41 & 0.673 \\
\bottomrule
\end{tabular}

\endgroup
\end{table}

The success controlled comparison in Table~\ref{tab:success-matched} gives
the same conclusion. At \(\alpha=0.75\), the Wilson lower limit is 96.8
percent for \PSAPST\ with 6.65 colors per vertex. Calibrated APST reaches a
99.1 percent lower limit with 13.15 colors per vertex. Uniform list greedy
does not attain the target on the tested scale grid.

\begin{figure}[H]
\centering
\includegraphics[width=0.82\textwidth]{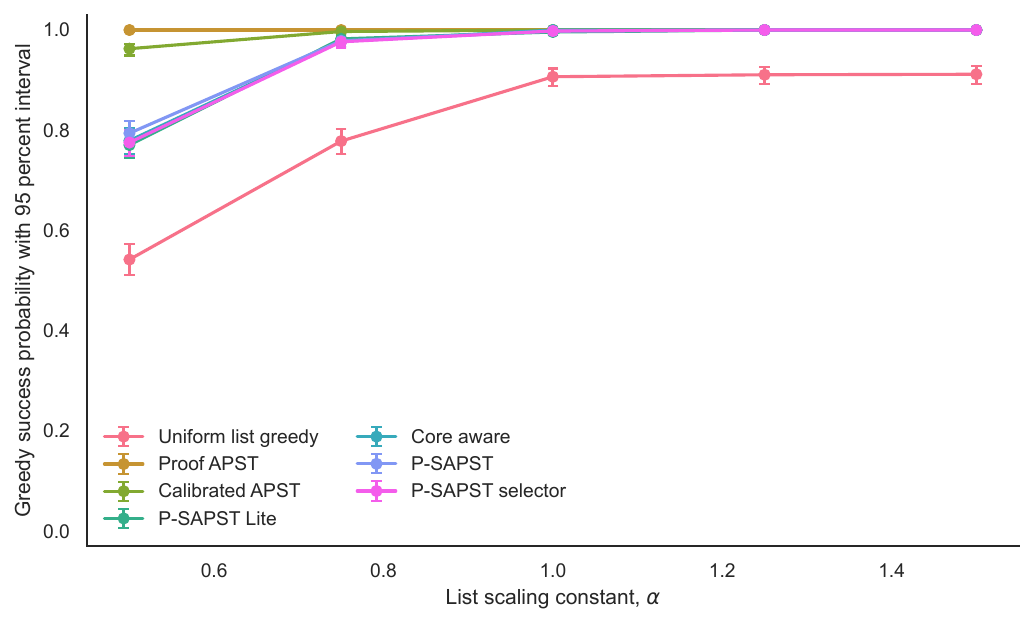}
\caption{Greedy success across list scales with 95 percent Wilson
intervals.}
\label{fig:success-scale}
\end{figure}

Figure~\ref{fig:success-scale} shows the transition from aggressive
compression to stable recovery. Proof APST saturates many palettes.
Calibrated and structure aware methods reach high success with much smaller
lists, while the common uniform budget remains sensitive to heterogeneous
degrees.

\begin{figure}[H]
\centering
\includegraphics[width=0.75\textwidth]{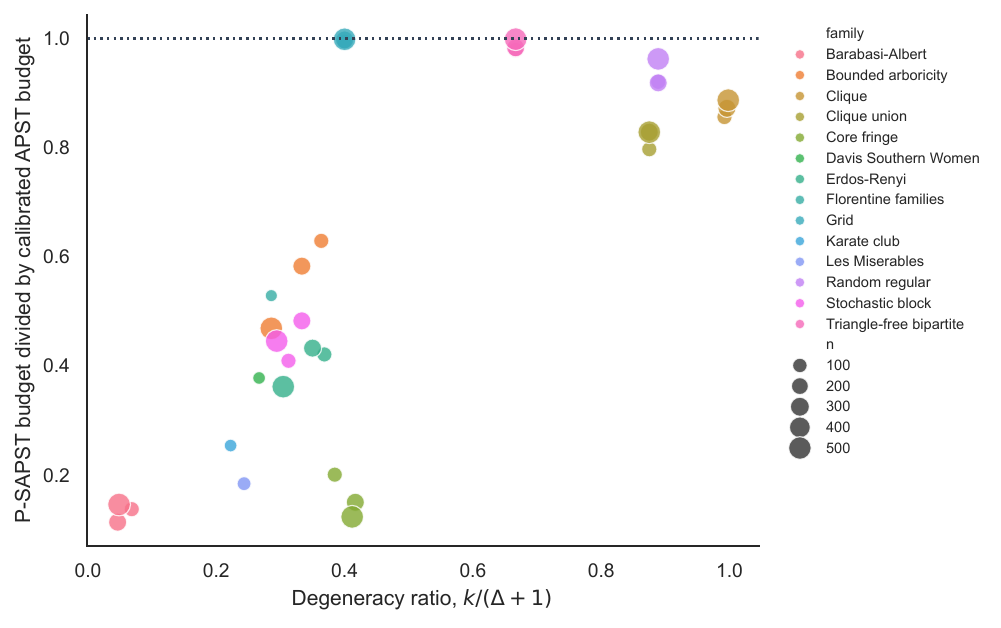}
\caption{\PSAPST\ budget divided by calibrated APST budget against the
degeneracy ratio \(k/(\Delta+1)\). Point area represents graph size.}
\label{fig:gap-budget}
\end{figure}

The structural mechanism is visible in Figure~\ref{fig:gap-budget}. The
largest gains occur when degeneracy is small relative to the palette.
Points approach or exceed a ratio of one when the profile remains dense,
which agrees with \cref{cor:dense-barrier}.

\begin{figure}[H]
\centering
\includegraphics[width=0.95\textwidth]{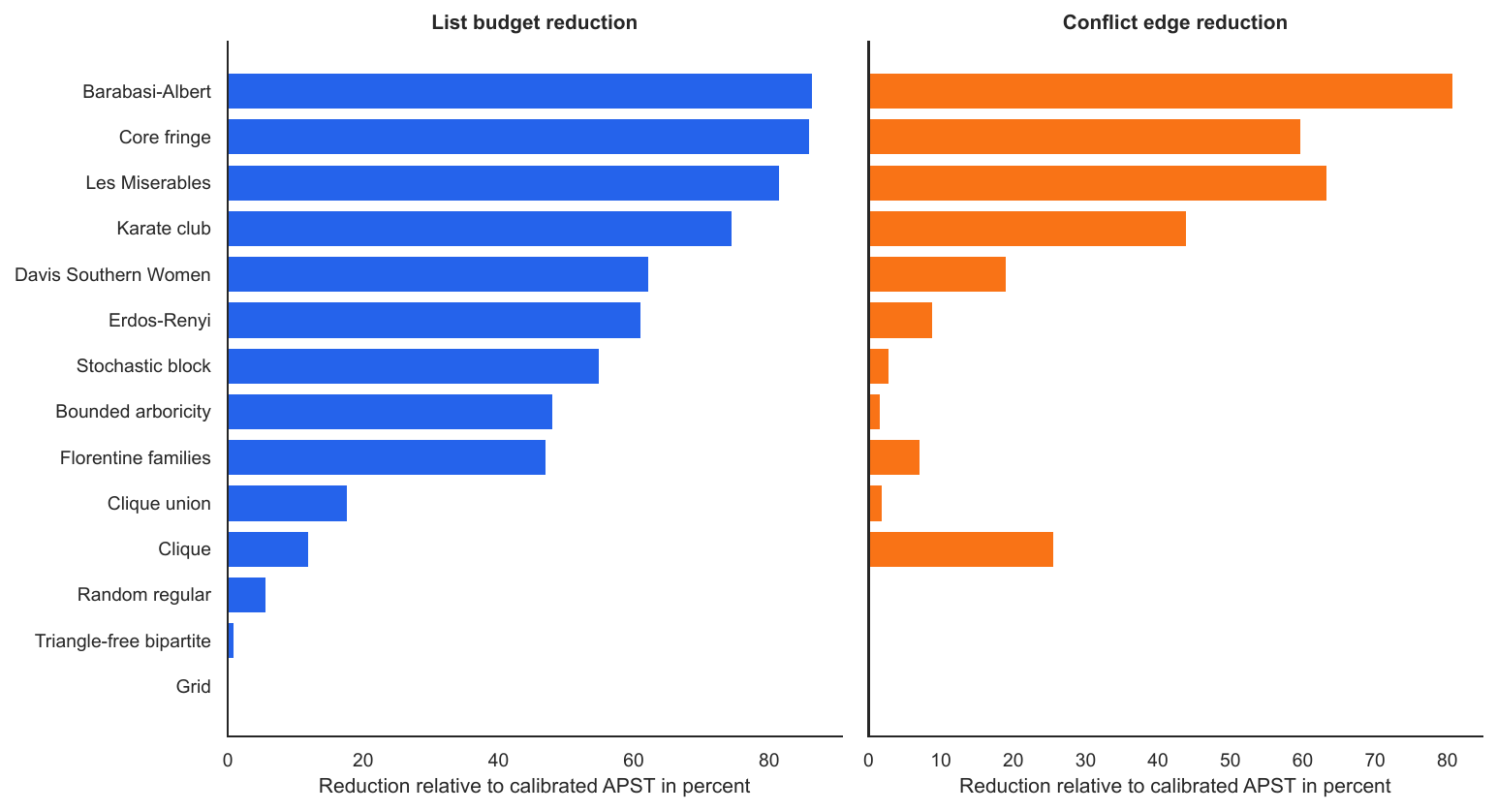}
\caption{Relative list and conflict reduction from calibrated APST by graph family at \(\alpha=1\). Positive values indicate compression by \PSAPST.}
\label{fig:family-reduction}
\end{figure}

Figure~\ref{fig:family-reduction} separates the structural regimes hidden by an aggregate mean. Hub dominated and power law families obtain the largest list reduction, while regular, lattice, and clique families offer little or no advantage. The same ordering broadly appears for conflict edges, although the exact overlap probability also depends on the two endpoint budgets.

\begin{figure}[H]
\centering
\includegraphics[width=0.95\textwidth]{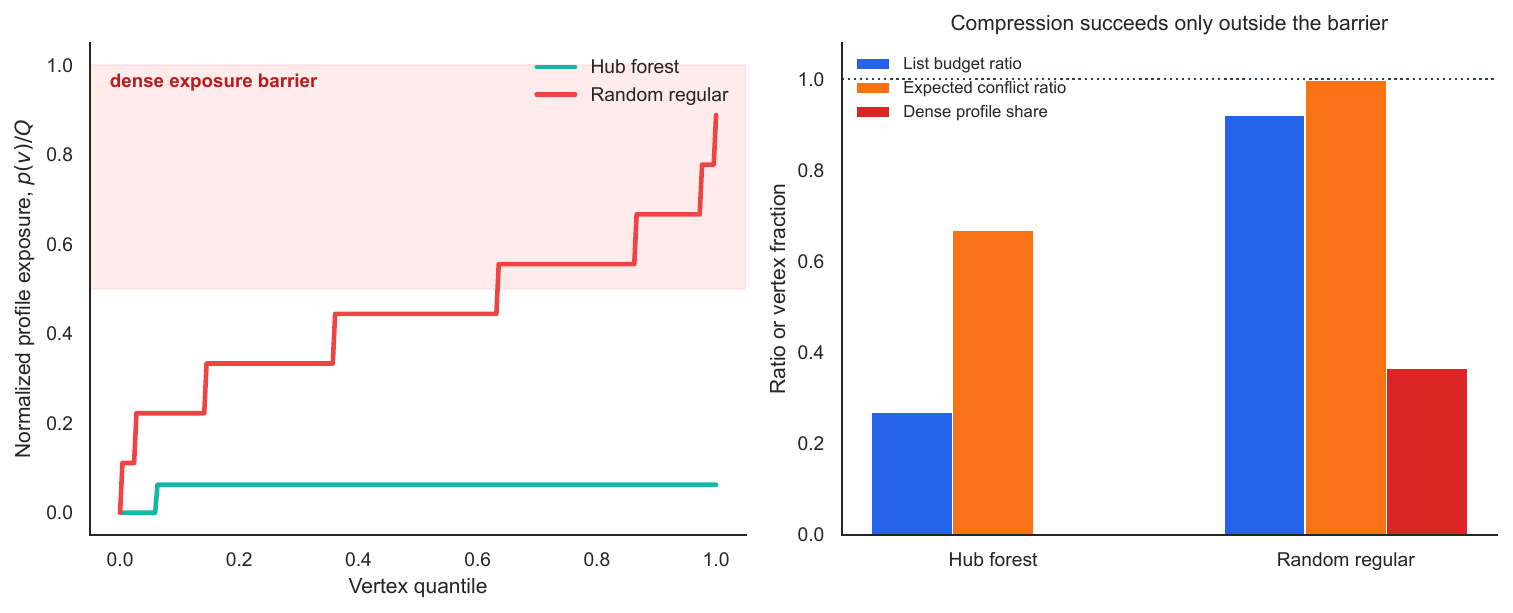}
\caption{A favorable hub forest and a failure case on an eight regular graph with 256 vertices. The normalized profile is sparse on the hub forest, so \PSAPST\ compresses both lists and conflicts. The regular graph has a dense exposure share near one, and its ratios approach the calibrated APST baseline, visually confirming the dense exposure barrier.}
\label{fig:barrier-failure}
\end{figure}

Figure~\ref{fig:barrier-failure} makes the dense exposure barrier explicit. The hub forest has budget ratio 0.270 and no vertex above half palette exposure. The random regular graph has budget ratio 0.918, expected conflict ratio 0.998, and 37.5 percent of vertices in the displayed dense region. The regular case therefore records a controlled failure of compression rather than an omitted unfavorable instance.

\begin{figure}[H]
\centering
\includegraphics[width=0.88\textwidth]{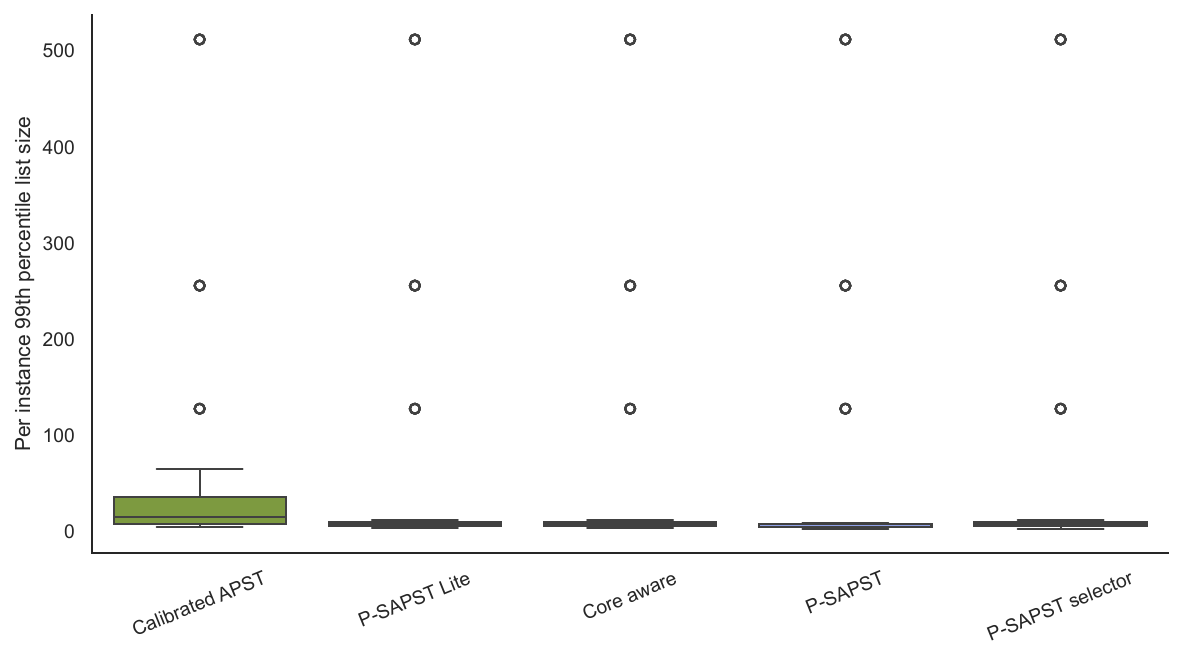}
\caption{Distribution of the per plan 99th percentile list size at
\(\alpha=1\).}
\label{fig:list-tail}
\end{figure}

Figure~\ref{fig:list-tail} reports the list tail requested by the local storage analysis. \PSAPST\ improves substantially over reciprocal rank allocation, but P-SAPST Lite and core orders remain competitive. The spread also records the dense families where all local rules approach the palette cap.

\subsection{Exact profile lower bound audit}

\begin{figure}[H]
\centering
\includegraphics[width=0.9\textwidth]{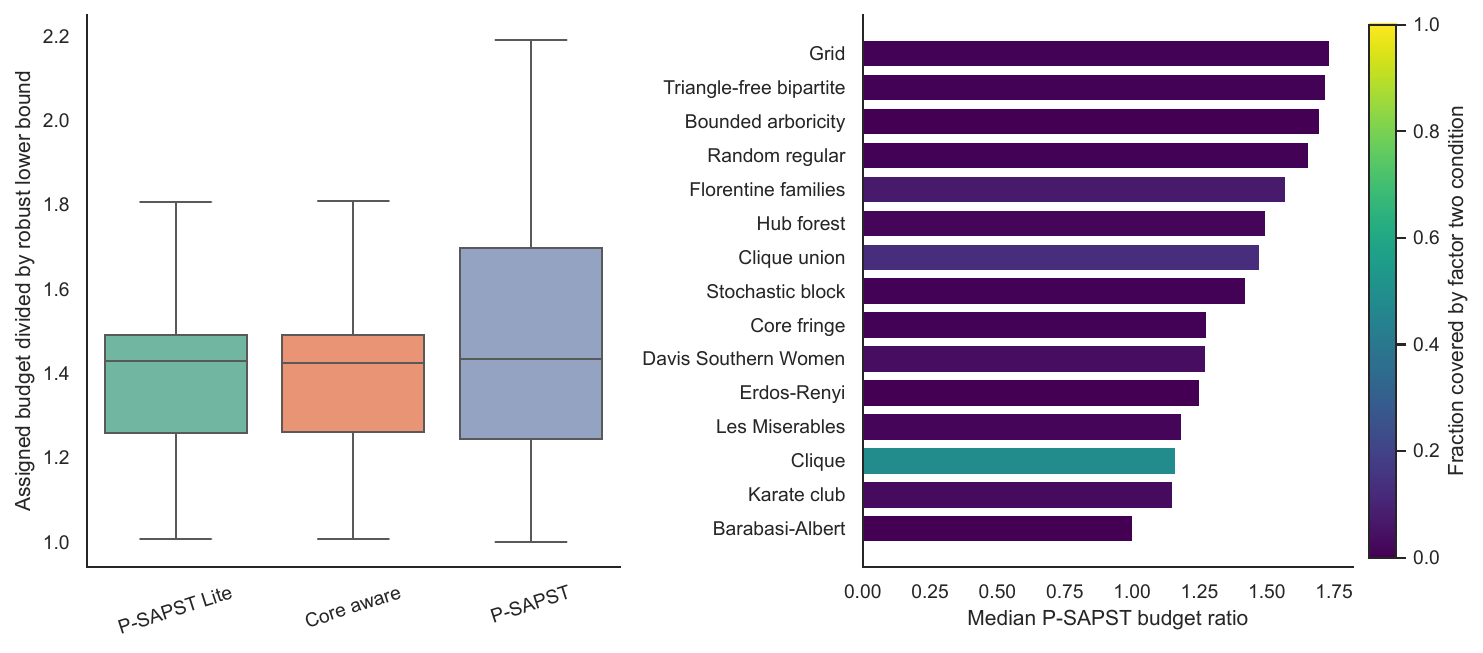}
\caption{Allocated budget divided by the exact profile lower bound \(H(\sigma)\) from \cref{thm:exact-local-optimum}. The ratio is computed before random list sampling.}
\label{fig:profile-audit}
\end{figure}

\begin{table}[H]
\centering
\caption{Budget ratio to the exact history robust profile optimum at \(\alpha=1\). The certified column records the share of vertices satisfying the simpler analytic factor two condition.}
\label{tab:profile-audit}
\begingroup
\setlength{\tabcolsep}{4pt}
\begin{tabular}{llrrrr}
\toprule
Suite & Method & Median ratio & P90 ratio & Certified & Instances \\
\midrule
Real networks & Core aware & 1.21 & 1.57 & 9.8\% & 4 \\
Real networks & P-SAPST & 1.23 & 1.57 & 3.5\% & 4 \\
Real networks & P-SAPST Lite & 1.21 & 1.57 & 9.6\% & 4 \\
SAPBench Core & Core aware & 1.42 & 1.57 & 12.3\% & 30 \\
SAPBench Core & P-SAPST & 1.43 & 1.73 & 6.5\% & 30 \\
SAPBench Core & P-SAPST Lite & 1.43 & 1.57 & 12.2\% & 30 \\
SAPBench Scale & Core aware & 1.48 & 1.77 & 8.1\% & 24 \\
SAPBench Scale & P-SAPST & 1.67 & 1.80 & 0.2\% & 24 \\
SAPBench Scale & P-SAPST Lite & 1.48 & 1.77 & 8.1\% & 24 \\
\bottomrule
\end{tabular}

\endgroup
\end{table}

The audit in Figure~\ref{fig:profile-audit} and Table~\ref{tab:profile-audit} compares each structural order with the strongest lower bound justified by the same exposure profile. For full \PSAPST, the median ratio is 1.43 on SAPBench Core and 1.67 on SAPBench Scale, with respective 90th percentiles 1.73 and 1.80. The Barab\'asi and Albert family has median ratio 1.00. These values establish near optimal calibration only within the history robust independent uniform list model. They do not constitute a lower bound for correlated sampling, non greedy recovery, or a different order.

\subsection{Conflict prediction}

\begin{figure}[H]
\centering
\includegraphics[width=\textwidth]{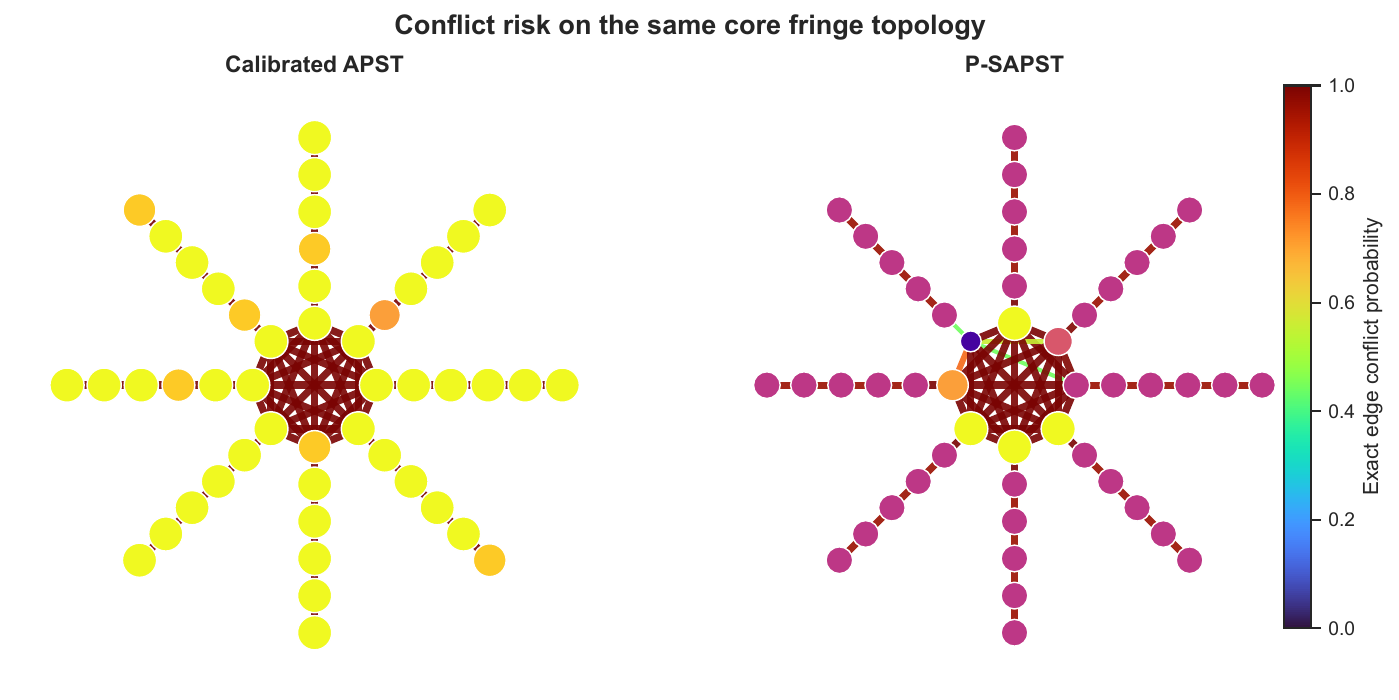}
\caption{Topology level conflict mechanism. Vertex color represents normalized local exposure, and edge color represents the exact probability that the endpoint lists intersect. The left panel uses reciprocal rank allocation, while the right panel uses the peeling profile.}
\label{fig:topology-conflict}
\end{figure}

Figure~\ref{fig:topology-conflict} locates the conflict potential instead of reporting only a graph level total. Reciprocal rank allocation distributes large budgets according to position and leaves substantial overlap on exposed edges. Profile calibration contracts lists on the sparse fringe and preserves larger lists only near the dense center. This spatial view explains why a single mean list size cannot identify where the compressed instance remains difficult.

\begin{figure}[H]
\centering
\includegraphics[width=0.95\textwidth]{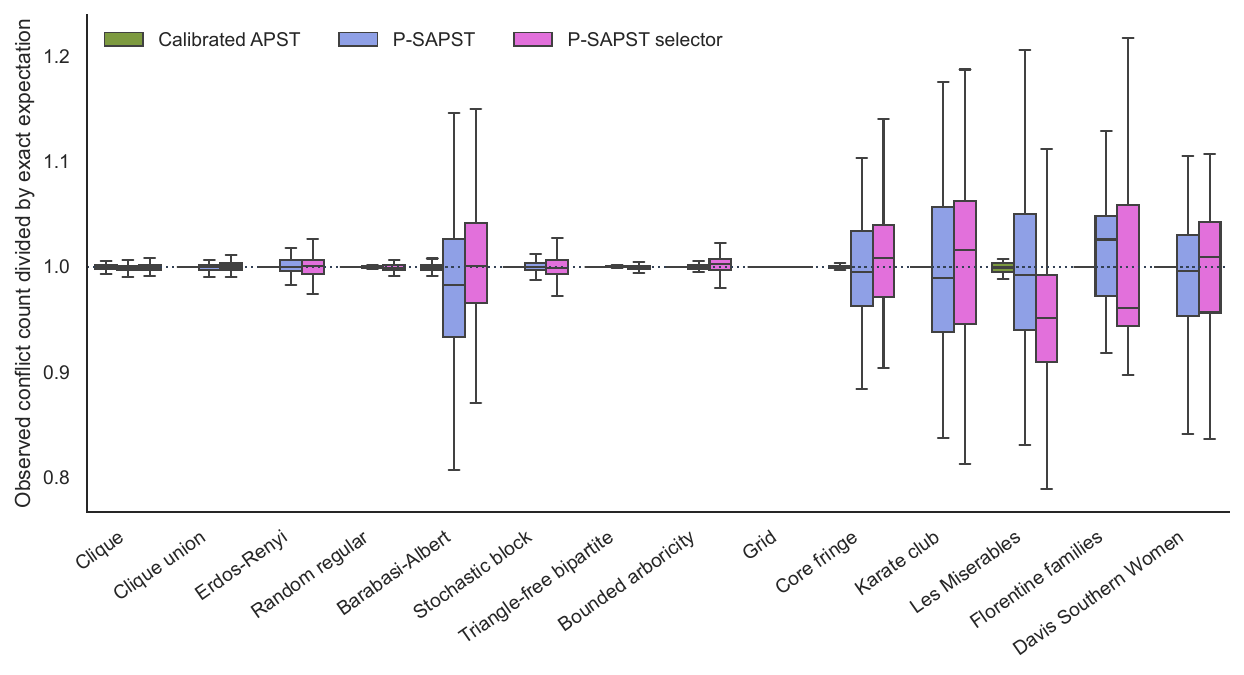}
\caption{Observed conflict count divided by the exact expectation. The
dotted line marks equality.}
\label{fig:conflict-residual}
\end{figure}

The median residual in Figure~\ref{fig:conflict-residual} is one for
calibrated APST, \PSAPST, and the selector. For \PSAPST, the fifth and 95th
percentiles are 0.934 and 1.060. This agreement validates the
hypergeometric expectation without reusing it as an observed performance
measure. Dependence among adjacent edge indicators creates variation, and
\cref{eq:conflict-concentration} controls that variation through the input
degrees.

\subsection{Scaling and phase costs}

\begin{figure}[H]
\centering
\includegraphics[width=\textwidth]{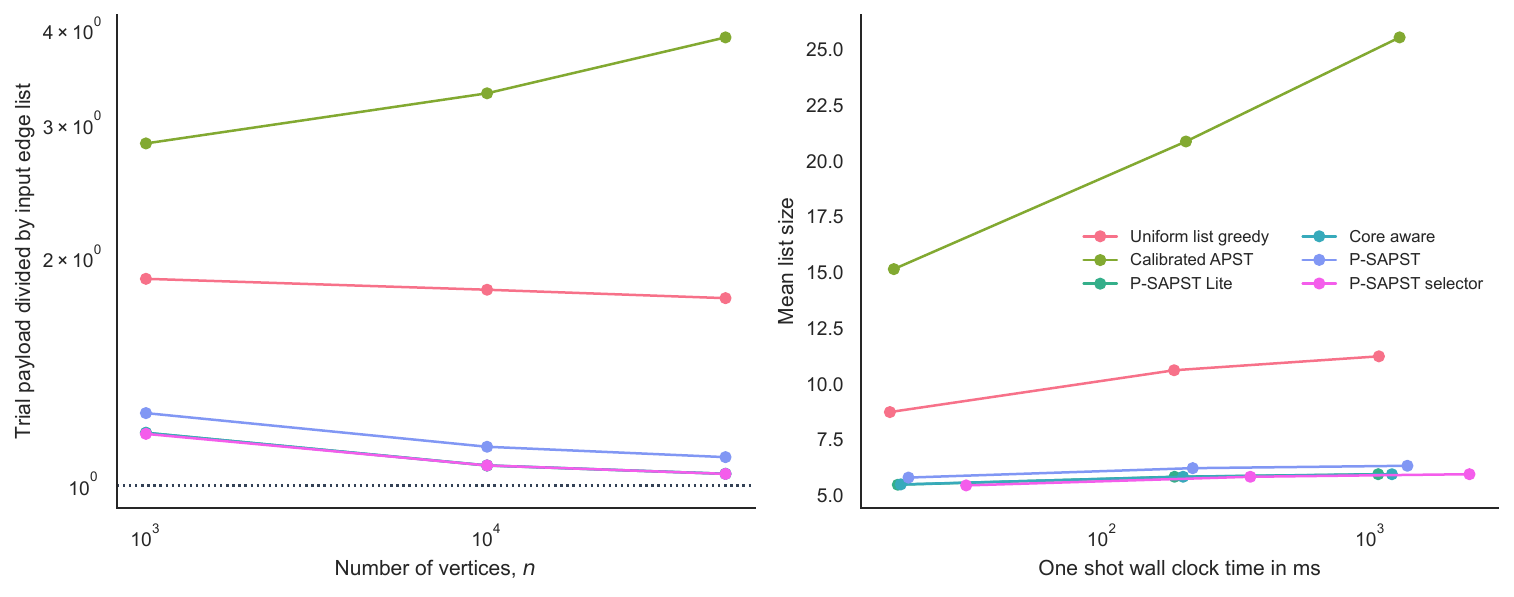}
\caption{Payload scaling and the running time versus list size tradeoff on
SAPBench Scale.}
\label{fig:scale-payload}
\end{figure}

\begin{table}[H]
\centering
\caption{SAPBench Scale at \(\alpha=1\). Each row aggregates eight graph
families and thirty independent list trials per family.}
\label{tab:scale-results}
\begingroup
\setlength{\tabcolsep}{4pt}
\begin{tabular}{lrrrrr}
\toprule
Method and size & Mean & P99 & Success & Payload & ms \\
\midrule
Core aware $n=1,000$ & 5.51 & 7.62 & 100.0\% & 1.174 & 17.71 \\
Core aware $n=10,000$ & 5.85 & 7.62 & 100.0\% & 1.063 & 199.81 \\
Core aware $n=50,000$ & 5.96 & 8.12 & 100.0\% & 1.036 & 1199.85 \\
P-SAPST Lite $n=1,000$ & 5.49 & 7.50 & 100.0\% & 1.174 & 17.26 \\
P-SAPST Lite $n=10,000$ & 5.85 & 7.88 & 100.0\% & 1.062 & 185.66 \\
P-SAPST Lite $n=50,000$ & 5.96 & 8.12 & 100.0\% & 1.036 & 1066.64 \\
Calibrated APST $n=1,000$ & 15.16 & 24.62 & 100.0\% & 2.833 & 16.69 \\
Calibrated APST $n=10,000$ & 20.88 & 61.75 & 100.0\% & 3.300 & 204.76 \\
Calibrated APST $n=50,000$ & 25.55 & 138.50 & 100.0\% & 3.912 & 1281.44 \\
P-SAPST $n=1,000$ & 5.81 & 6.50 & 99.2\% & 1.246 & 18.92 \\
P-SAPST $n=10,000$ & 6.23 & 6.88 & 100.0\% & 1.125 & 217.05 \\
P-SAPST $n=50,000$ & 6.34 & 6.88 & 100.0\% & 1.090 & 1370.18 \\
P-SAPST selector $n=1,000$ & 5.46 & 7.50 & 100.0\% & 1.170 & 31.04 \\
P-SAPST selector $n=10,000$ & 5.85 & 7.88 & 100.0\% & 1.063 & 355.91 \\
P-SAPST selector $n=50,000$ & 5.96 & 8.12 & 100.0\% & 1.036 & 2334.18 \\
Uniform list greedy $n=1,000$ & 8.75 & 8.75 & 99.6\% & 1.876 & 16.13 \\
Uniform list greedy $n=10,000$ & 10.62 & 10.62 & 99.6\% & 1.814 & 185.03 \\
Uniform list greedy $n=50,000$ & 11.25 & 11.25 & 100.0\% & 1.769 & 1072.95 \\
\bottomrule
\end{tabular}

\endgroup
\end{table}

At 50,000 vertices, \PSAPST\ uses 6.34 colors per vertex and has an average
99th percentile of 6.88. Calibrated APST uses 25.55 colors per vertex and
has a 99th percentile of 138.50. Their payload ratios are 1.090 and 3.912.
The corresponding mean times are 1.37 s and 1.28 s, so the compressed
instance is substantially smaller without an end to end speedup at this
size.

The advantage is concentrated in the intended structural regime. On the
50,000 vertex hub forest, \PSAPST\ uses 2.99 colors per vertex and retains
3.9 percent of the input edges, while calibrated APST uses 52.83 colors and
retains 94.6 percent. On the power law graph of the same size, the means are
3.00 and 70.29. In contrast, grids, regular graphs, and regular bipartite
graphs retain nearly every edge because their small palettes and sustained
exposure force large local lists.

\begin{figure}[H]
\centering
\includegraphics[width=\textwidth]{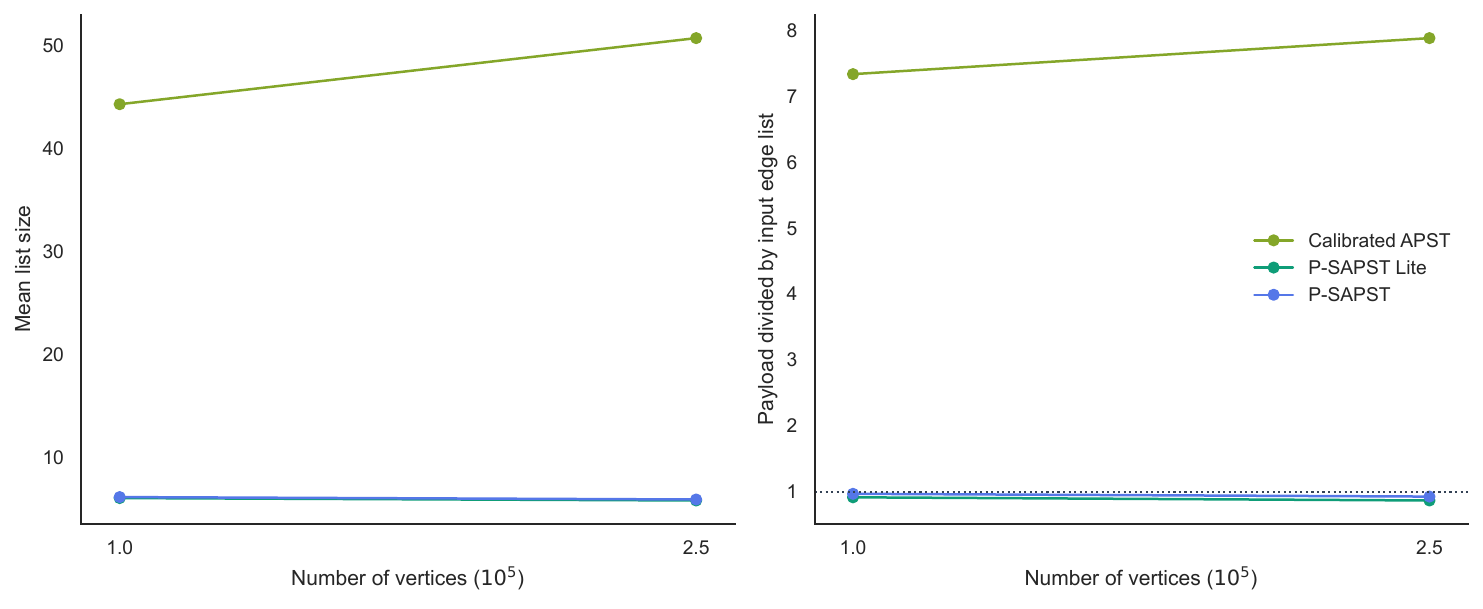}
\caption{SAPBench Stress at 100,000 and 250,000 vertices. The panels report list volume, payload ratio, and end to end running time.}
\label{fig:large-scale-stress}
\end{figure}

\begin{table}[H]
\centering
\caption{SAPBench Stress. Edge counts and performance values are averaged over four graph families. Each method uses ten independent list samples per instance, and every run succeeds.}
\label{tab:large-scale-stress}
\begingroup
\setlength{\tabcolsep}{3pt}
\begin{tabular}{lrrrrrr}
\toprule
Method and size & Edges & Mean & P99 & Conflict & Payload & ms \\
\midrule
P-SAPST Lite $n=100,000$ & 299,878 & 6.09 & 8.75 & 0.456 & 0.914 & 2522 \\
P-SAPST Lite $n=250,000$ & 750,337 & 5.85 & 8.50 & 0.441 & 0.865 & 6434 \\
Calibrated APST $n=100,000$ & 299,878 & 44.35 & 372.50 & 0.893 & 7.339 & 3617 \\
Calibrated APST $n=250,000$ & 750,337 & 50.77 & 522.75 & 0.860 & 7.886 & 10147 \\
P-SAPST $n=100,000$ & 299,878 & 6.18 & 7.25 & 0.505 & 0.967 & 3477 \\
P-SAPST $n=250,000$ & 750,337 & 5.95 & 6.75 & 0.497 & 0.925 & 9294 \\
\bottomrule
\end{tabular}

\endgroup
\end{table}

The stress results in Figure~\ref{fig:large-scale-stress} and Table~\ref{tab:large-scale-stress} extend the implementation to 250,000 vertices and more than 1.25 million edges. At the larger size, P-SAPST Lite uses 5.85 colors per vertex, retains 44.1 percent of edges, and has payload ratio 0.865. Calibrated APST uses 50.77 colors, retains 86.0 percent of edges, and has payload ratio 7.886. Full \PSAPST\ uses 5.95 colors with payload ratio 0.925. All 240 stress runs complete greedy recovery. The aggregate again hides sharp differences. On the 250,000 vertex power law graph, Lite uses 2.99 colors and payload ratio 0.310, whereas on the bounded arboricity construction it uses 8.51 colors and payload ratio 1.271.

\subsection{SNAP network results}

\begin{table}[H]
\centering
\caption{Results on two SNAP networks at \(\alpha=1\). The column \(\Delta/k\) reports maximum degree and degeneracy. Every method uses ten independent list samples, and all runs succeed.}
\label{tab:real-networks}
\begingroup
\setlength{\tabcolsep}{2.6pt}
\begin{tabular}{llrrrrrrr}
\toprule
Network & Method & $n$ & $m$ & $\Delta/k$ & Mean & Conflict & Payload & ms \\
\midrule
ca-HepTh & P-SAPST Lite & 9,877 & 25,973 & 65/31 & 3.90 & 0.278 & 0.649 & 342 \\
ca-HepTh & Calibrated APST & 9,877 & 25,973 & 65/31 & 33.19 & 0.998 & 4.154 & 282 \\
ca-HepTh & P-SAPST & 9,877 & 25,973 & 65/31 & 3.94 & 0.307 & 0.682 & 369 \\
email-Enron & P-SAPST Lite & 36,692 & 183,831 & 1383/43 & 2.69 & 0.009 & 0.193 & 1050 \\
email-Enron & Calibrated APST & 36,692 & 183,831 & 1383/43 & 76.51 & 0.564 & 5.814 & 3577 \\
email-Enron & P-SAPST & 36,692 & 183,831 & 1383/43 & 2.69 & 0.009 & 0.194 & 1035 \\
\bottomrule
\end{tabular}

\endgroup
\end{table}

Table~\ref{tab:real-networks} confirms that the synthetic stress trend is present in observed heterogeneous networks. On ca HepTh, \PSAPST\ uses 3.94 colors per vertex and payload ratio 0.682, while calibrated APST uses 33.19 colors and payload ratio 4.154. On email Enron, \PSAPST\ uses 2.69 colors and payload ratio 0.194, compared with 76.51 colors and 5.814 for calibrated APST. P-SAPST Lite is marginally smaller on both networks. The difference between maximum degree and degeneracy is especially large on email Enron, which is the structural regime predicted by the profile analysis.

\begin{figure}[H]
\centering
\includegraphics[width=0.85\textwidth]{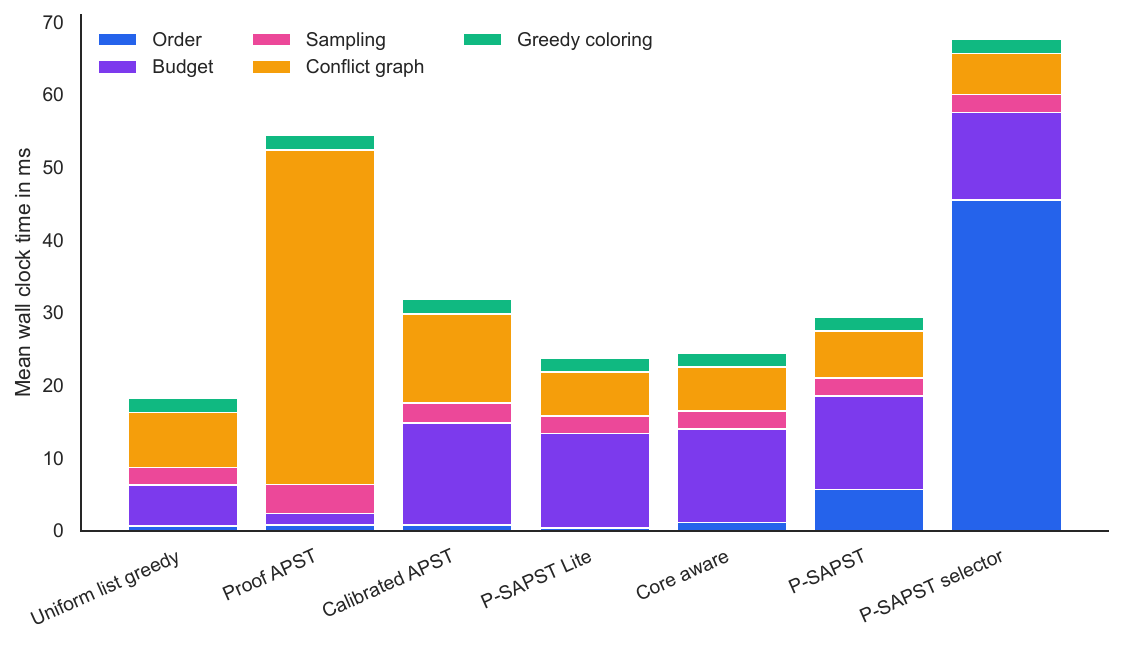}
\caption{Mean phase level wall clock time on SAPBench Core at
\(\alpha=1\).}
\label{fig:runtime-breakdown}
\end{figure}

\begin{table}[H]
\centering
\caption{Mean phase times in milliseconds on SAPBench Core at
\(\alpha=1\).}
\label{tab:runtime-breakdown}
\begingroup
\setlength{\tabcolsep}{4pt}
\begin{tabular}{lrrrrr}
\toprule
Method & Order & Budget & Sampling & Conflict & Coloring \\
\midrule
Uniform list greedy & 0.70 & 5.65 & 2.41 & 7.54 & 1.94 \\
Proof APST & 0.80 & 1.59 & 3.99 & 46.02 & 2.00 \\
Calibrated APST & 0.81 & 14.05 & 2.75 & 12.26 & 1.98 \\
P-SAPST Lite & 0.40 & 13.02 & 2.44 & 5.98 & 1.93 \\
Core aware & 1.16 & 12.89 & 2.48 & 6.03 & 1.89 \\
P-SAPST & 5.72 & 12.89 & 2.42 & 6.50 & 1.91 \\
P-SAPST selector & 45.55 & 12.03 & 2.46 & 5.65 & 1.93 \\
\bottomrule
\end{tabular}

\endgroup
\end{table}

Figure~\ref{fig:runtime-breakdown} and
Table~\ref{tab:runtime-breakdown} isolate the planning cost. \PSAPST\ spends
5.72 ms on order construction, compared with 0.81 ms for calibrated APST.
Its smaller conflict graph reduces the later conflict phase enough that the
mean end to end time is slightly lower on the core suite. The selector
spends 45.55 ms on order comparison and is slower overall.

When one structural plan is reused for \(r\) independent list samples, the average time per compressed instance is
\begin{equation}
A_r
=
\frac{T_{\mathrm{plan}}}{r}
+
T_{\mathrm{budget}}
+
T_{\mathrm{sample}}
+
T_{\mathrm{conflict}}
+
T_{\mathrm{color}}.
\label{eq:amortized-time}
\end{equation}
where \(T_{\mathrm{plan}}\) is paid once and every other term is paid per sample.

\begin{figure}[H]
\centering
\includegraphics[width=0.88\textwidth]{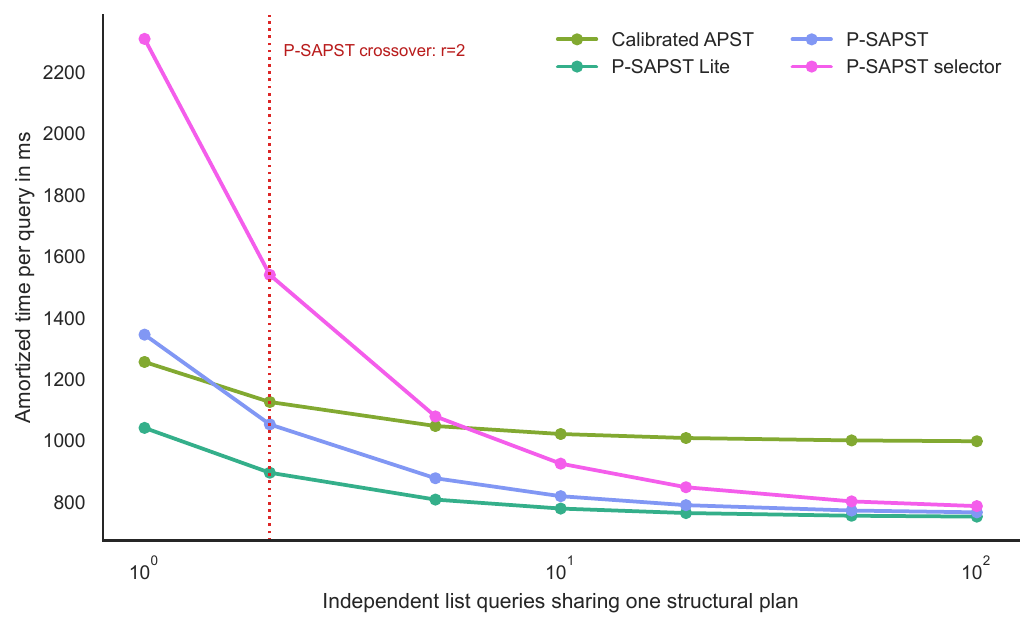}
\caption{Mean time per compressed instance when a graph dependent plan is reused across independent list samples.}
\label{fig:amortized-reuse}
\end{figure}

Figure~\ref{fig:amortized-reuse} instantiates \cref{prop:repeated-query-crossover} using the measured phase decomposition at 50,000 vertices. The suite averaged crossover for full \PSAPST\ occurs at the second independent query. At \(r=10\), the amortized times are 821 ms for \PSAPST, 780 ms for P-SAPST Lite, and 1,023 ms for calibrated APST. The selector crosses after six queries and remains relevant only when its modest compression gain justifies repeated reuse.

\begin{figure}[H]
\centering
\includegraphics[width=0.72\textwidth]{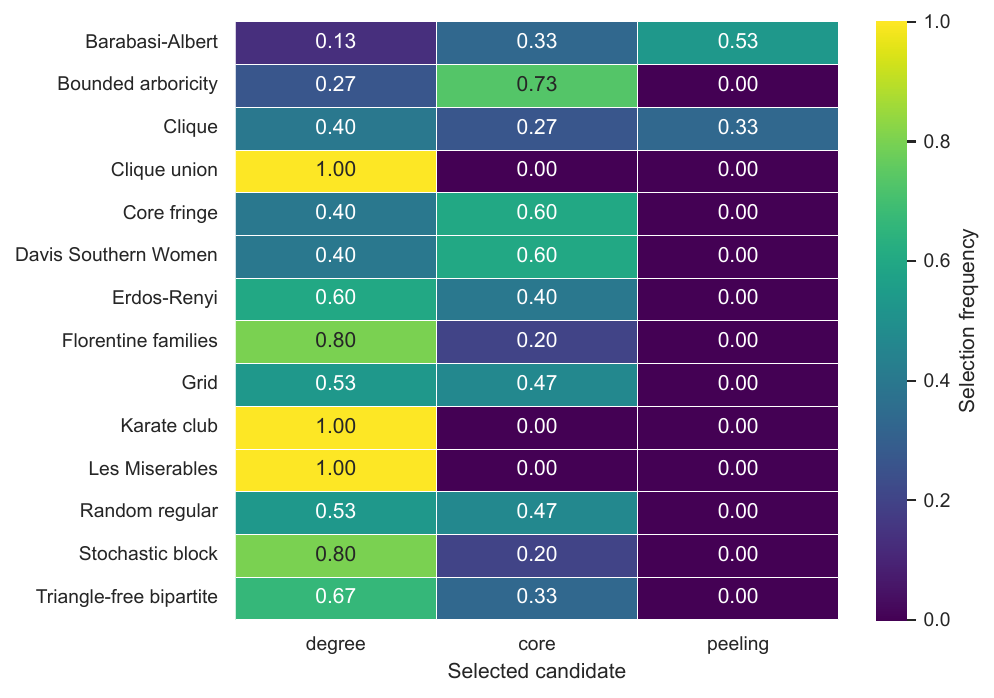}
\caption{Order selected by the optional three candidate extension at \(\alpha=1\). The selector optimizes exact expected conflict size and therefore often favors degree or core order. Peeling remains the certified candidate because it alone supplies the pointwise profile identity.}
\label{fig:selector-map}
\end{figure}

The selector map confirms that peeling is chosen infrequently under the exact conflict objective. Degree is selected most often, followed by core order. This preference concerns the empirical objective in \cref{eq:selector-objective}; it does not provide a worst case certificate for the chosen order. Peeling remains the theorem bearing instance because \cref{thm:profile-identity} identifies every exposure before sampling, while the selector is an engineering extension for repeated settings where the planning cost can be amortized.

\subsection{Limitations}

The experiments do not establish a universal ordering advantage. Dense exposure can force lists close to the full palette, and the compact payload can still exceed the input edge list. P-SAPST Lite has lower planning cost but lacks an unconditional relation to the peeling profile on adversarial graphs. The method assumes that adjacency is available before list sampling, so it does not directly apply to an edge arrival stream. A semi streaming extension would require a compact sketch that estimates a useful removal or exposure profile without storing the full graph. The fixed profile lower bound applies only to independent uniform lists with history robust greedy recovery. Correlated sampling, nonuniform colors, and recovery with revisions may escape it because the conditional product argument no longer applies. The real network study contains two SNAP graphs and is not a comprehensive survey of social, web, road, or biological collections. SAPBench Stress establishes larger implementation scale, while broader external evaluation remains future work.

\FloatBarrier

\section{Conclusion}
\label{sec:conclusion}

\PSAPST\ is the certified degeneracy guided instance of a broader exposure calibrated ordering framework. Reverse minimum degree peeling turns residual degrees into exact backward exposures, which then determine vertexwise list budgets before any color is sampled. Within the independent uniform history robust model, the resulting profile has an exact local benchmark and yields strict separations on high degree forests and core fringe graphs. Exact conflict bounds and the dense exposure barrier identify both the compressible region and its boundary. P-SAPST Lite uses the same calibration with a cheaper degree order and has a conditional transfer guarantee when its exposure profile is dominated by peeling. The experiments find payload ratios below one on hub dominated synthetic graphs and both SNAP networks, while regular and dense families provide explicit failure cases.

The primary open problem is to approximate a useful exposure profile in the semi streaming model. A successful sketch would need to preserve enough of the removal structure to retain local budget gains without materializing adjacency. Further questions include lower bounds for correlated lists or non greedy recovery, dynamic maintenance of approximate profiles, extensions to variable list and defective coloring, and faster order rules with a provable relation to the best exposure calibrated order.

\appendix
\section{Proofs}
\label{sec:appendix-proofs}

\begin{proof}[Proof of \cref{lem:local-calibration}]
Condition on the graph, the auxiliary randomness, and the resulting plan.
When greedy recovery reaches \(v_i\), let \(F_i\) be the colors assigned to
its earlier neighbors. The recovery history determines \(F_i\), and
\begin{equation}
\lvert F_i\rvert
\leq
b_\sigma(v_i).
\label{eq:forbidden-size}
\end{equation}
where repeated neighbor colors can only make the forbidden set smaller.

By \cref{eq:history-independence}, the current list remains a uniform subset
independent of \(F_i\). If \(\ell(v_i)>b_\sigma(v_i)\), failure is
impossible. Otherwise,
\begin{align}
\Prb\!\left(
L(v_i)\subseteq F_i
\mid\mathcal{H}_i
\right)
&\leq
\frac{\binom{b_\sigma(v_i)}{\ell(v_i)}}
{\binom{\Q}{\ell(v_i)}}\notag\\
&=
\prod_{j=0}^{\ell(v_i)-1}
\frac{b_\sigma(v_i)-j}{\Q-j}\notag\\
&\leq
\left(
\frac{b_\sigma(v_i)}{\Q}
\right)^{\ell(v_i)}
\leq
\eta_{v_i}.
\label{eq:conditional-failure}
\end{align}
where the ratio is the exact probability that sampling without replacement
selects only forbidden colors. The last inequality follows from
\cref{eq:local-calibration}. If the uncapped expression exceeds \(\Q\),
the assigned list is the full palette and failure is zero.

A union bound over the conditional failure events gives
\begin{equation}
\Prb(\text{greedy recovery fails})
\leq
\sum_{v\in V}\eta_v
\leq
\delta.
\label{eq:global-failure}
\end{equation}
where the first inequality does not require independence among the failure
events.
\end{proof}

\begin{proof}[Proof of \cref{thm:profile-identity}]
Fix \(r_i\). At the moment it is removed, its active neighbors are exactly
\begin{equation}
N(r_i)\cap\{r_{i+1},\ldots,r_n\}.
\label{eq:active-neighbors}
\end{equation}
where these vertices occur before \(r_i\) after the removal sequence is
reversed. Vertices \(r_1,\ldots,r_{i-1}\) occur after \(r_i\) and therefore
do not contribute to its backward neighborhood. Hence
\(B_{\sigma_{\mathrm{peel}}}(r_i)\) equals the set in
\cref{eq:active-neighbors}, whose cardinality is \(p(r_i)\).
\end{proof}

\begin{proof}[Proof of \cref{thm:forest-separation}]
Every nonempty forest has a vertex of degree at most one. Minimum degree
peeling therefore satisfies
\begin{equation}
p(v)\leq 1
\qquad
\text{for every }v\in V.
\label{eq:forest-profile}
\end{equation}
where isolated vertices have profile value zero. For a vertex with
\(p(v)=1\),
\begin{align}
\lambda_{\delta/n}(1)
&=
\min\!\left\{
\Q,
\left\lceil
\frac{\log(n/\delta)}{\log\Q}
\right\rceil
\right\}\notag\\
&\leq
\left\lceil
\frac{(1+c)\log n}{\gamma\log n}
\right\rceil
=
\left\lceil
\frac{1+c}{\gamma}
\right\rceil.
\label{eq:forest-local-budget}
\end{align}
where \(\delta=n^{-c}\) and \(\Q\geq n^\gamma\). Summing proves
\cref{eq:forest-budget}.

A forest has fewer than \(n\) edges. Applying
\cref{eq:conflict-upper} with the bound in
\cref{eq:forest-local-budget} gives
\begin{equation}
\E\lvert E_L\rvert
\leq
\frac{n}{\Delta+1}\,
\left\lceil\frac{1+c}{\gamma}\right\rceil^2.
\label{eq:forest-conflict-proof}
\end{equation}
where the right side is \(O(n/\Delta)\).

For the reciprocal rank rule, let
\begin{equation}
t=
\left\lceil
\frac{c_0n\log n}{\Q}
\right\rceil.
\label{eq:rank-threshold}
\end{equation}
where \(t<n\) for sufficiently large \(n\) because
\(\Q\geq n^\gamma\). The saturated prefix contributes at most
\(O(n\log n)\), and the remaining harmonic tail is at most
\(O(n\log^2 n)\). For the converse asymptotic estimate,
\begin{align}
\sum_{j=t}^{n}\ell_R(j)
&\geq
c_0n\log n
\sum_{j=t}^{n}\frac{1}{j}\notag\\
&=
\Omega\!\left(
n\log n\,
\log\frac{n}{t}
\right)
=
\Omega(n\log^2 n).
\label{eq:rank-tail}
\end{align}
where
\(\log(n/t)=\Omega(\log(\Q/\log n))=\Omega(\log n)\).
This proves \cref{eq:rank-forest-budget}.
\end{proof}

\begin{proof}[Proof of \cref{thm:core-fringe-separation}]
Every pendant vertex is removed before the clique core and has removal degree one. After the fringe is removed, the residual graph is \(K_s\), so each core removal degree is at most \(s-1\). The palette size is
\begin{equation}
\Q=s+t.
\label{eq:core-fringe-palette}
\end{equation}
where the maximum degree is \(s+t-1\). For a fringe exposure \(b=1\), the assumptions give \(\Q/b\geq t\geq n^\gamma\). For a core exposure \(1\leq b\leq s-1\),
\begin{equation}
\frac{\Q}{b}
\geq
\frac{t}{s}
\geq
n^\gamma.
\label{eq:core-fringe-gap}
\end{equation}
where the final inequality is the fringe dominance condition. Therefore every positive profile value receives at most \(C_{\gamma,c}\) colors, while a zero profile value receives one color. Summing proves \cref{eq:core-fringe-budget}.

The graph has
\begin{equation}
m=st+\binom{s}{2}\leq s(s+t)=s\Q.
\label{eq:core-fringe-edges}
\end{equation}
where the first term counts fringe edges and the second counts core edges. Applying \cref{eq:conflict-upper} with endpoint budgets at most \(C_{\gamma,c}\) proves \cref{eq:core-fringe-conflict}.

Finally, \(t\geq s n^\gamma\) and \(n=s(t+1)\) imply that \(\Q\) grows polynomially with \(n\), while \(\Q\leq n\). The reciprocal rank harmonic tail argument from the proof of \cref{thm:forest-separation} therefore applies unchanged and proves \cref{eq:core-fringe-rank-budget}.
\end{proof}

\begin{proof}[Proof of \cref{prop:lite-transfer}]
Fix a vertex with \(p(v)>0\). The exposure dominance condition gives
\begin{align}
\log\frac{\Q}{b_{\mathrm{L}}(v)}
&\geq
\log\frac{\Q}{\Q^{1-\theta}p(v)^\theta}\notag\\
&=
\theta\log\frac{\Q}{p(v)}.
\label{eq:lite-transfer-log}
\end{align}
where the right side is a fraction \(\theta\) of the peeling exposure potential. If the peeling budget is not capped, dividing \(\log(1/\eta)\) by \cref{eq:lite-transfer-log} and rounding proves \cref{eq:lite-budget-transfer}. If the peeling budget is capped at \(\Q\), the right side of \cref{eq:lite-budget-transfer} is at least \(\Q\), while every list is capped at \(\Q\). Vertices with \(p(v)=0\) have zero Lite exposure by assumption and receive one color under both orders.
\end{proof}

\begin{proof}[Proof of \cref{thm:exact-local-optimum}]
Fix a forbidden set \(F\subseteq[\Q]\) with \(\lvert F\rvert=b\). For a
uniform list of size \(\ell\), the failure probability is
\begin{equation}
\Prb(L\subseteq F)
=
\frac{\binom{b}{\ell}}{\binom{\Q}{\ell}}.
\label{eq:exact-optimum-proof}
\end{equation}
where the ratio is zero when \(\ell>b\). The value depends only on \(b\),
not on the identity of \(F\). Therefore a list is valid for every
independent forbidden set of size \(b\) exactly when the ratio is at most
\(\eta\). The smallest such integer is \(h_\eta(b)\), which proves
\cref{eq:exact-optimum-condition}. Applying the same necessary condition to
every vertex and summing proves \cref{eq:profile-optimum}.
\end{proof}

\begin{proof}[Proof of \cref{prop:repeated-query-crossover}]
Subtracting the two repeated query costs gives
\begin{equation}
T_M(r)-T_R(r)
=
P_M-P_R-r(q_R-q_M).
\label{eq:repeated-query-difference}
\end{equation}
where \(q_R-q_M>0\) by assumption. The right side is negative under the condition in \cref{eq:repeated-query-threshold}, which proves the claim.
\end{proof}

\begin{proof}[Proof of \cref{thm:local-lower}]
For a forbidden set of size exactly \(b\), the failure probability is
\begin{equation}
\frac{\binom{b}{\ell}}{\binom{\Q}{\ell}}
=
\prod_{j=0}^{\ell-1}
\frac{b-j}{\Q-j}.
\label{eq:exact-local-failure}
\end{equation}
where the expression is valid for \(\ell\leq b\). If
\(\ell\leq b/2\), then \(b-j\geq b/2\) and \(\Q-j\leq\Q\) throughout the
product. Consequently,
\begin{equation}
\frac{\binom{b}{\ell}}{\binom{\Q}{\ell}}
\geq
\left(\frac{b}{2\Q}\right)^\ell.
\label{eq:local-failure-lower}
\end{equation}
where the right side is a lower bound on the worst history failure
probability. Requiring this quantity to be at most \(\eta\) proves
\cref{eq:local-lower-bound}.

If \(b\leq\Q/2\), then
\begin{equation}
\log(2\Q/b)
=
\log(\Q/b)+\log 2
\leq
2\log(\Q/b).
\label{eq:factor-two-denominator}
\end{equation}
where the final inequality uses \(\log(\Q/b)\geq\log 2\). Thus the lower
bound is at least half of the unrounded sufficient value in
\cref{eq:local-calibration}.
\end{proof}

\begin{proof}[Proof of \cref{cor:dense-barrier}]
Fix \(v\in H_\beta\), and write \(b=b_\sigma(v)\). If
\(\ell(v)>b/2\), then
\begin{equation}
\ell(v)
\geq
\left\lfloor\frac{\beta\Q}{2}\right\rfloor+1.
\label{eq:dense-large-case}
\end{equation}
where \(b\geq\beta\Q\). Otherwise, \cref{thm:local-lower} gives
\begin{equation}
\ell(v)
\geq
\frac{\log(1/\eta)}{\log(2\Q/b)}
\geq
\frac{\log(1/\eta)}{\log(2/\beta)}.
\label{eq:dense-small-case}
\end{equation}
where the second inequality again uses \(b\geq\beta\Q\). Therefore every
vertex in \(H_\beta\) has list size at least \(s_\beta\), which proves
\cref{eq:dense-storage-lower}. The function
\(\rho_{\Q}(x,y)\) is nondecreasing in both arguments, so each edge induced
by \(H_\beta\) has intersection probability at least
\(\rho_{\Q}(s_\beta,s_\beta)\). Linearity of expectation proves
\cref{eq:dense-conflict-lower}.
\end{proof}

\begin{proof}[Proof of \cref{thm:conflict-theory}]
Fix an edge \(\{u,v\}\), and abbreviate its list sizes by \(x\) and \(y\).
After \(L(u)\) is fixed, the number of \(y\)-subsets disjoint from it is
\(\binom{\Q-x}{y}\) when \(x+y\leq\Q\). Therefore
\begin{equation}
\Prb(L(u)\cap L(v)\neq\varnothing)
=
\rho_{\Q}(x,y).
\label{eq:edge-intersection-proof}
\end{equation}
where the probability is one if the two list sizes sum to more than the
palette size. Summing edge indicators proves
\cref{eq:exact-conflict-expectation}. Independence among edge indicators is
not needed.

The lists form independent vertex variables. Replacing \(L(v)\) can alter
only the conflict indicators on edges incident to \(v\), so the bounded
difference constant is \(d(v)\). McDiarmid's inequality yields
\cref{eq:conflict-concentration}.

Finally, a union bound over the \(xy\) cross pairs of list positions gives
\begin{equation}
\rho_{\Q}(x,y)
\leq
\min\!\left\{1,\frac{xy}{\Q}\right\}.
\label{eq:intersection-upper-proof}
\end{equation}
where each position pair matches with probability \(1/\Q\). Summing over
the input edges proves \cref{eq:conflict-upper}.
\end{proof}

\begin{proof}[Proof of \cref{prop:complexity}]
The heap stores one current degree key per active vertex. Removing a vertex
and updating its active neighbors performs \(O(n+m)\) heap operations, each
with logarithmic cost. A scan of the reversed order computes all profile
budgets. Sampling touches each stored color once. Greedy recovery scans the
lists and adjacency structure, while direct conflict construction tests
each input edge through its smaller endpoint list. These operations give
the stated time bounds and require the adjacency representation together
with the sampled lists.
\end{proof}

\begingroup
\setlength{\bibsep}{0pt plus 0.2ex}
\bibliographystyle{plainnat}
\bibliography{references}

@inproceedings{assadi2019sublinear,
  title={Sublinear algorithms for ($\Delta$+ 1) vertex coloring},
  author={Assadi, Sepehr and Chen, Yu and Khanna, Sanjeev},
  booktitle={Proceedings of the Thirtieth Annual ACM-SIAM Symposium on Discrete Algorithms},
  pages={767--786},
  year={2019},
  organization={SIAM}
}

@article{assadi2026simple,
  title={Simple Sublinear Algorithms for {$(\Delta+1)$} Vertex Coloring via Asymmetric Palette Sparsification},
  author={Assadi, Sepehr and Yazdanyar, Helia},
  journal={TheoretiCS},
  volume={5},
  year={2026},
  publisher={Episciences. org}
}

@article{kahn2023asymptotics,
  title={Asymptotics for palette sparsification},
  author={Kahn, Jeff and Kenney, Charles},
  journal={arXiv preprint arXiv:2306.00171},
  year={2023}
}

@article{kahn2024asymptotics,
  title={Asymptotics for palette sparsification from variable lists},
  author={Kahn, Jeff and Kenney, Charles},
  journal={arXiv preprint arXiv:2407.07928},
  year={2024}
}

@article{ashvinkumar2024palette,
  title={Palette Sparsification via FKNP},
  author={Ashvinkumar, Vikrant and Kenney, Charles},
  journal={arXiv preprint arXiv:2408.12835},
  year={2024}
}

@article{hefetz2025coloring,
  title={Coloring Graphs From Random Lists},
  author={Hefetz, Dan and Krivelevich, Michael},
  journal={Random Structures \& Algorithms},
  volume={66},
  number={1},
  pages={e21278},
  year={2025},
  publisher={Wiley Online Library}
}

@article{alon2020palette,
  title={Palette sparsification beyond {$(\Delta+1)$} vertex coloring},
  author={Alon, Noga and Assadi, Sepehr},
  journal={arXiv preprint arXiv:2006.10456},
  year={2020}
}

@article{dhawan2024palette,
  title={Palette Sparsification for Graphs with Sparse Neighborhoods},
  author={Dhawan, Abhishek},
  journal={arXiv preprint arXiv:2408.08256},
  year={2024}
}

@inproceedings{flin2024distributed,
  title={A distributed palette sparsification theorem},
  author={Flin, Maxime and Ghaffari, Mohsen and Halld{\'o}rsson, Magn{\'u}s M and Kuhn, Fabian and Nolin, Alexandre},
  booktitle={Proceedings of the 2024 Annual ACM-SIAM Symposium on Discrete Algorithms (SODA)},
  pages={4083--4123},
  year={2024},
  organization={SIAM}
}

@inproceedings{chang2019complexity,
  title={The complexity of ($\Delta$+ 1) coloring in congested clique, massively parallel computation, and centralized local computation},
  author={Chang, Yi-Jun and Fischer, Manuela and Ghaffari, Mohsen and Uitto, Jara and Zheng, Yufan},
  booktitle={Proceedings of the 2019 ACM symposium on principles of distributed computing},
  pages={471--480},
  year={2019}
}

@inproceedings{halldorsson2022near,
  title={Near-optimal distributed degree+ 1 coloring},
  author={Halld{\'o}rsson, Magn{\'u}s M and Kuhn, Fabian and Nolin, Alexandre and Tonoyan, Tigran},
  booktitle={Proceedings of the 54th Annual ACM SIGACT Symposium on Theory of Computing},
  pages={450--463},
  year={2022}
}

@article{maus2023distributed,
  title={Distributed graph coloring made easy},
  author={Maus, Yannic},
  journal={ACM Transactions on Parallel Computing},
  volume={10},
  number={4},
  pages={1--21},
  year={2023},
  publisher={ACM New York, NY}
}

@article{feigenbaum2005graph,
  title={On graph problems in a semi-streaming model},
  author={Feigenbaum, Joan and Kannan, Sampath and McGregor, Andrew and Suri, Siddharth and Zhang, Jian},
  journal={Theoretical Computer Science},
  volume={348},
  number={2-3},
  pages={207--216},
  year={2005},
  publisher={Elsevier}
}

@article{bera2019graph,
  title={Graph coloring via degeneracy in streaming and other space-conscious models},
  author={Bera, Suman K and Chakrabarti, Amit and Ghosh, Prantar},
  journal={arXiv preprint arXiv:1905.00566},
  year={2019}
}

@inproceedings{assadi2022deterministic,
  title={Deterministic graph coloring in the streaming model},
  author={Assadi, Sepehr and Chen, Andrew and Sun, Glenn},
  booktitle={Proceedings of the 54th Annual ACM SIGACT Symposium on Theory of Computing},
  pages={261--274},
  year={2022}
}

@inproceedings{assadi2023coloring,
  title={Coloring in graph streams via deterministic and adversarially robust algorithms},
  author={Assadi, Sepehr and Chakrabarti, Amit and Ghosh, Prantar and Stoeckl, Manuel},
  booktitle={Proceedings of the 42nd ACM SIGMOD-SIGACT-SIGAI Symposium on Principles of Database Systems},
  pages={141--153},
  year={2023}
}

@article{matula1983smallest,
  title={Smallest-last ordering and clustering and graph coloring algorithms},
  author={Matula, David W and Beck, Leland L},
  journal={Journal of the ACM (JACM)},
  volume={30},
  number={3},
  pages={417--427},
  year={1983},
  publisher={ACM New York, NY, USA}
}

@article{erdds1959random,
  title={On random graphs I},
  author={ERDdS, PARW and R\&wi, A},
  journal={Publ. math. debrecen},
  volume={6},
  number={290-297},
  pages={18},
  year={1959}
}

@article{barabasi1999emergence,
  title={Emergence of scaling in random networks},
  author={Barab{\'a}si, Albert-L{\'a}szl{\'o} and Albert, R{\'e}ka},
  journal={science},
  volume={286},
  number={5439},
  pages={509--512},
  year={1999},
  publisher={American Association for the Advancement of Science}
}

@article{holland1983stochastic,
  title={Stochastic blockmodels: First steps},
  author={Holland, Paul W and Laskey, Kathryn Blackmond and Leinhardt, Samuel},
  journal={Social networks},
  volume={5},
  number={2},
  pages={109--137},
  year={1983},
  publisher={Elsevier}
}

@techreport{hagberg2007exploring,
  title={Exploring network structure, dynamics, and function using NetworkX},
  author={Hagberg, Aric and Swart, Pieter J and Schult, Daniel A},
  year={2007},
  institution={Los Alamos National Laboratory (LANL)}
}

@article{leskovec2016snap,
  title={Snap: A general-purpose network analysis and graph-mining library},
  author={Leskovec, Jure and Sosi{\v{c}}, Rok},
  journal={ACM Transactions on Intelligent Systems and Technology (TIST)},
  volume={8},
  number={1},
  pages={1--20},
  year={2016},
  publisher={ACM New York, NY, USA}
}
\endgroup

\end{document}